\documentclass[10pt]{amsart}
\usepackage{lineno}

\usepackage{todonotes}
\usepackage{soul}

\usepackage[margin=2.2cm]{geometry}
\usepackage{fullpage}
\usepackage{tikz}
\usepackage{tikz-cd}
\usetikzlibrary{matrix,arrows,shapes,cd}
\usepackage{pgf}
\usepackage{bbold}
\usepackage{url}
\usepackage{longtable}
\usepackage{enumerate}
\usepackage{amsthm}
\usepackage{amssymb}
\usepackage{amsmath}
\usepackage{amscd}
\usepackage{eucal}
\usepackage{mathrsfs}
\usepackage{upgreek}
\usepackage{mathtools}
\usepackage{dsfont}
\usepackage{yfonts}
\usepackage[T1]{fontenc}
\usepackage{bbm}
\usepackage{graphics}
\usepackage{array}
\usepackage[retainorgcmds]{IEEEtrantools}
\usepackage{booktabs}
\usepackage{enumerate}
\usepackage{feynmf}
\usepackage{amsaddr}

\usetikzlibrary{trees}
\usetikzlibrary{decorations.pathmorphing}
\usetikzlibrary{decorations.markings}
\tikzset{
	photon/.style={decorate, decoration={snake}, draw=red},
	electron/.style={draw=blue, postaction={decorate},
		decoration={markings,mark=at position .55 with {\arrow[draw=blue]{>}}}},
	gluon/.style={decorate, draw=magenta,
		decoration={coil,amplitude=4pt, segment length=5pt}},
	sderiv/.style={postaction={decorate},
		decoration={markings,mark=at position .3 with {\arrow{>}}}},
	tderiv/.style={postaction={decorate},
		decoration={markings,mark=at position .7 with {\arrow{<}}}},
	stderiv/.style={postaction={decorate},
		decoration={markings,mark=at position .7 with {\arrow{<}},mark=at position .3 with {\arrow{>}}}}
}
\definecolor{see}{RGB}{67,75,179}
\definecolor{darksee}{RGB}{42,44,148}
\definecolor{honey}{RGB}{232,180,129}
\definecolor{lighthoney}{RGB}{255,254,220}
\definecolor{citecol}{rgb}{0.5,0,0} 
\definecolor{blue1}{RGB}{130,150,209}

\DeclareSymbolFont{bbold}{U}{bbold}{m}{n}
\DeclareSymbolFontAlphabet{\mathbbold}{bbold}

\newcommand{\pol}{\mathrm{pol}}
\usepackage{soul}

\setstcolor{red}
\definecolor{see}{RGB}{67,75,179}
\usepackage{hyperref}
\hypersetup
{
	bookmarksopen=true,
	pdfmenubar=true, 
	pdfhighlight=/O, 
	colorlinks=true, 
	pdfpagelayout=SinglePage, 
	pdffitwindow=true, 
	linkcolor=see, 
	citecolor=red!75!darkgray, 
	urlcolor=see 
}
\newcommand{\fA}{\mathfrak{A}}

\newcommand{\euD}{\mathscr{D}}  
  
\newcommand{\Lcal}{\mathcal {L}}

\newcommand{\Dcal}{\mathcal{D}}
\newcommand{\Ecal}{\mathcal{E}} 

\newcommand{\Fcal}{\mathcal{F}} 
\newcommand{\Acal}{\mathcal{A}}

\newcommand{\Ocal}{\mathcal{O}}
\newcommand{\Scal}{\mathcal{S}}

\newcommand{\Rcal}{\mathcal{R}}
\newcommand{\Tcal}{\mathcal{T}}
\newcommand{\Vcal}{\mathcal{V}}

\newcommand{\Zcal}{\mathcal{Z}}

\newcommand{\Qcal}{\mathcal{Q}}

\newcommand{\Ci}{C^\infty} 
\newcommand{\Hom}{\mathrm{Hom}}

\newcommand{\Obs}{\mathrm{\mathbf{Obs}}}       

\newcommand{\Sym}{\mathrm{Sym}}     

\newcommand{\WF}{\mathrm{WF}}         
\newcommand{\id}{\mathrm{id}}               
\newcommand{\dvol}{\mathrm{dvol}} 
\DeclareMathOperator{\tr}{\mathrm{tr}}                 

\newcommand{\loc}{\mathrm{loc}}

\newcommand{\mc}{{\mu\mathrm{c}}}

\newcommand{\NN}{\mathbb{N}}          
\newcommand{\RR}{\mathbb{R}}           
\newcommand{\CC}{\mathbb{C}}           

\newcommand{\Tb}{\mathbb{T}}

\newcommand{\al}{\alpha}

\newcommand{\ph}{\phi}

\newcommand{\T}{\cdot_{{}^\Tcal}}

\newcommand{\TT}{\Tcal}

\newcommand{\vr}[1]{\boldsymbol{#1}}         
\newcommand{\be}{\begin{equation}}
\newcommand{\ee}{\end{equation}}

\newcommand{\Lap}{\bigtriangleup}

\DeclareMathOperator{\supp}{\mathrm{supp}}      

\newcommand{\Pei}[2]{\lfloor #1, #2 \rfloor}

\theoremstyle{plain}
\newtheorem{thm}{Theorem}[section]
\newtheorem{df}[thm]{Definition}
\newtheorem{df-thm}[thm]{Definition/Theorem}
\newtheorem{prop}[thm]{Proposition}
\newtheorem{cor}[thm]{Corollary}
\newtheorem{lemma}[thm]{Lemma}

\theoremstyle{definition}
\newtheorem{rem}[thm]{Remark}

\newtheorem{exa}{Example}
\newtheorem*{NB}{Nota Bene}

\def\d{{\rm d}}
\def\xto{\xrightarrow}

\def\Ch{\mathbf{Ch}}

\def\Vec{\mathbf{Vec}}

\def\lag{{\rm Lag}} 
\def\llag{{{\mathbb L} {\rm ag}}}
\def\loc{{\rm Loc}} 
\def\lloc{{{\mathbb L} {\rm oc}}}
\def\mloc{{\rm MLoc}} 
\def\mlloc{{{\mathbb M}{\mathbb L} {\rm oc}}}

\def\cinfty{C^\infty}
\def\Obs{{\rm Obs}}
\def\Opens{{\rm Opens}}
\def\Dens{{\rm Dens}}
\def\cDens{\Dcal{\rm ens}}
\def\jet{{\rm j}}

\title{The observables of a perturbative algebraic quantum field theory form a factorization algebra}

\author[1]{\small{Owen Gwilliam}}
\address{ University of Massachusetts, Amherst,\\ Department of Mathematics,\\
\normalfont{\texttt{gwilliam@math.umass.edu}}}
\author[2]{\small{Kasia Rejzner}}
\address{University of York,  \\
	Department of Mathematics,\\
\normalfont{\texttt{kasia.rejzner@york.ac.uk}}
}

\date{\today}

\begin{document}
 \sloppy

\maketitle
\begin{abstract}
We demonstrate that perturbative algebraic QFT methods, as developed by  Fredenhagen and Rejzner, naturally yields a factorization algebras of observables for a large class of Lorentzian theories. 
Along the way we carefully articulate cochain-level refinements of multilocal functionals, building upon results about the variational bicomplex, and we lift existing results about Epstein-Glaser renormalization to these multilocal differential forms, results
which may be of independent interest. 
\end{abstract}

\tableofcontents


 
\section{Introduction}

Despite its spectacular success as a framework for theoretical physics, the mathematical foundations of quantum field theory (QFT) are not entirely well-understood. While the non-perturbative construction of 4-dimensional interacting QFT models is yet to be achieved, perturbative QFT has lately received a lot of attention from mathematicians. One of the mathematical approaches, developed for QFT in Lorentzian signature, goes under the name \emph{perturbative algebraic quantum field theory} (pAQFT) \cite{BDF,BF0,BF97,BFBook,DF,DF04,Due19,Book}. Another approach, developed for QFT in Riemannian signature, uses factorization algebras \cite{Cos,CoGw,CG2}. Comparing and relating them is a natural problem, with hopes of producing a comprehensive view on perturbative field theories.

In the context of free theory, these two approaches were brought together by us in \cite{GR20}, using the Batalin-Vilkovisky (BV) formalism for quantization but without needing renormalization and without discussing the local-to-global property. 
The main insight of that paper was that the time-ordered product of pAQFT is, morally speaking, the factorization product of Costello-Gwilliam, when considered on the appropriate class of opens. 
In particular, we demonstrated that on such opens, the pAQFT dg algebra with the time-ordered product and the classical BV differential is quasi-isomorphic to the CG dg algebra with the classical product and the deformed (quantum) BV differential. 

In this paper we tackle interacting theories.
Our central goal is to demonstrate precisely the slogan that BV quantization of a classical field theory provides a deformation quantization of its factorization algebra of observables.
The essential challenge is that renormalization is necessary, so the techniques of \cite{GR20} are insufficient.

In practice we show that the pAQFT machinery produces not only a net of algebras on causally convex opens, 
but it in fact yields a natural factorization algebra.
This result has two appealing features from the perspective of contemporary mathematical physics:
\begin{itemize}
\item it demonstrates the unity of field theory, as a corresponding statement in Riemannian signature is the central result of \cite{CG2},
which uses a quite different approach to renormalization, and
\item it demonstrates that the net of a pAQFT satisfies an interesting local-to-global axiom, which has appeared---for independent reasons---in recent work in algebraic topology and geometric representation theory.
\end{itemize}
Any pAQFT already constructed --- see, for instance, \cite{BF97,DFqed,H,FR3,BRZ,BFRej13} --- thus has an associated factorization algebra.
In this paper we do not examine any particular examples in detail,
but we discuss the case of pure Yang-Mills theory.
We expect that there will be interesting consequences when one studies pAQFTs on manifolds with nontrivial topology. 

\begin{rem}
Another important and closely related problem is how factorization algebras relate to the nets of algebras appearing in the Haag-Kastler framework
(i.e., the structure of observables as functors on spacetimes) 
rather than on constructive methods (which is the focus of \cite{GR20} and of this paper).
See \cite{BPS19} for an insightful examination of this question.
Local-to-global aspects of AQFT have been pursued by Benini, Schenkel, and collaborators under the term \emph{homotopy AQFT} \cite{BS17,BSS17,BSW19,BS19},
and it is intriguing to wonder what our constructions and their formalism can add to one another.
\end{rem}

\subsection{The main theorems}

On the level of classical theory, 
our main result is that the space $\mloc$ of multilocal functionals, used in pAQFT as a model for classical observables, has the structure of a factorization algebra.
A multilocal functional is a polynomial built from local functionals (i.e., the integrals of compactly supported Lagrangian densities).
We find it convenient to use local differential forms as a resolution of local functionals,
where a local differential form is a differential form-valued Lagrangian (cf. a Lagrangian density).
There is a corresponding resolution of multilocal functionals as well that we dub {\em multilocal differential forms}, denoted $\mlloc$.
Setting up these refined structures is a one technical aspect of this paper.
We use this machinery to prove the following.

\begin{thm}
Given a classical BV field theory in the sense of \cite{FR} with Lagrangian $\Lcal$,
there exists a strict factorization algebra $\Obs^{cl}$ assigning to each open set $U$ in spacetime, 
the cochain complex $\mloc(U)$ of multilocal functionals with the differential~$\Qcal_S = \{S,-\}$.

There exists a homotopy factorization algebra assigning to each open set in spacetime, 
the cochain complex of multilocal differential forms $\mlloc(U)$ with a differential $\delta_S$ lifting~$\Qcal_S$. 
\end{thm}

If we restrict to the subcategory of causally convex balls inside all opens,
we obtain a version of the usual dg pAQFT model for the classical field theory.
This comparison is similar to that of~\cite{GR20}.

Note that pAQFT encompasses both Lorentzian and Riemannian signatures,
so this theorem offers a new approach to observables in Euclidean field theory.
There is a natural map of factorization algebras from the classical observables constructed here to the classical observables constructed in \cite{CG2},
extending the embedding of local functionals into the CG classical observables.
(This map can be seen, in some sense, as the universal Noether map, 
although it falls outside the hypotheses used in~\cite{CG2}.)
The observables used in this paper are close in spirit to the constructions of Beilinson and Drinfeld~\cite{BD};
it might be fruitful to move beyond a spiritual connection.

The main result on quantum theory is that the above structure can be deformed by using the methods pf pAQFT to provide a factorization algebra that characterizes quantum observables of the theory.
Such a deformation may not exist; an anomaly must vanish.
(See Section~\ref{sec obsq} for an explanation of what a BV quantization means in pAQFT and with our extension to multilocal differential forms.)

\begin{thm}
Given a classical BV field theory in the sense of \cite{FR} that admits a BV quantization in the sense of \cite{FR3},
there exists a strict factorization algebra $\Obs^{q}$ assigning to each open set in spacetime,
the cochain complex of multilocal functionals with values $\CC[[\hbar]]$ with the differential~$\hat{s}=\Qcal_S + \d_\mloc -i\hbar \Delta_{\lambda\Lcal_I}$.
\end{thm}

If we restrict to the subcategory of causally convex balls inside all opens,
we obtain a version of the usual dg pAQFT model for the quantum field theory.
The time-ordered product is related to the factorization product, as in~\cite{GR20}.

We do {\em not} construct a homotopy factorization algebra of quantum observables,
as it would require proving a version of the anomalous master Ward identity at the level of multilocal differential forms.
Such a result involves a number of new ideas and techniques,
and we hope to provide it in future work.

In Riemannian signature, this factorization algebra of quantum observables is different than that constructed in~\cite{CG2}.
It would be wonderful to know if the map of classical observables can be extended to a map of quantum observables.
The different methods of construction (notably with renormalization methods) makes a direct treatment rather challenging,
but one should expect there to be compatible deformations of the factorization algebras of classical observables.

\subsection{Overview of the paper}

Our paper has two parts, with first of independent interest. 

In the first part (sections \ref{sec:local}-\ref{sec:ren}) we describe local and multilocal functionals --- with small variations on preceding work --- and their differential form refinements.
We explain how these lead to
a factorization algebra.
The arguments are standard in the theory of manifolds, modestly stretched to cover classical field theory,
but, in fact, no dynamics is relevant here.
We then explain how renormalization \`a la Epstein-Glaser, 
as developed in \cite{EG,BDF,H,FR3} to apply to pAQFT,
lifts to our multilocal differential forms.

In the second part (sections \ref{sec:classobs}-\ref{sec obsq}),
we incorporate dynamics by equipping this factorization algebra with a differential determined by the action functional;
the differential is local in spacetime and hence it continues to satisfy the cosheaf condition.
We call this dg factorization algebra the {\em classical observables} (section~\ref{sec:classobs}).
We then turn to the quantum differential in pAQFT,
leading to the {\em quantum observables} (section \ref{sec obsq}).
This deformation of the classical differential to the quantum differential is a key feature of the BV formalism,
and it leads to a deformation of the factorization algebra of classical observables to the factorization algebra of quantum observables.

We remark that the quantum differential is local in nature only if the quantized action satisfies the quantum master equation (QME),
and so satisfying the QME is necessary to obtain this deformation of factorization algebras. Possible obstructions to satisfy QME are identified as \emph{anomalies}.

\subsection{Outlook}

The results and techniques of this paper open up several avenues of exploration.

First, our result suggests a potential refinement of AQFT --- a local-to-global principle --- as we see that the examples of pAQFT are factorization algebras.
This aspect is most pertinent on spacetime manifolds with nontrivial topology.
(Note that to produce a classical BV field theory in our framework requires writing down an ``extended'' Lagrangian density, i.e., possibly adjoining Lagrangian with values in lower dimensional forms, satisfying the classical master equation with the GS bracket.
This issue becomes important only when the manifold is not contractible.)
Moreover, there are a number of factorization algebras arising from algebraic topology and representation theory, 
and it is interesting to ask if they might play a role in AQFT.

Second, AQFT has a wealth of ideas and constructions (such as superselection sectors, the theory of phases, DHR reconstruction) 
that ought to admit analogues for the factorization algebras of Euclidean field theories, 
and perhaps for other factorization algebras.
Our result suggests examples by which to explore these analogies.

Third, our use of local differential forms is closely related to their use in the BFV formalism \cite{CMR1,CMR2,mnev2020towards} and in higher form symmetries \cite{gaiotto2015generalized}.
We want to explore how these areas intertwine with AQFT and factorization algebras.

\subsection{Acknowledgements}

Our views on field theory are deeply shaped by Klaus Fredenhagen and Kevin Costello,
so their vision is reflected in our work here even if they played no direct role in it;
we offer gratitude for their guidance and mentorship in mathematical physics.
Eli Hawkins and Berend Visser have read drafts and offered invaluable feedback.
We have both benefited from conversations with Michele Schiavina, particularly about local differential forms and their role in BV formalism and its BFV cousins.
O.G. would like to thank Urs Schreiber and Igor Khavkine for pressing him on these issues and discussing how BV theories behave in Lorentzian signature at a nice dinner long ago in Hamburg.
Finally, we thank the referees, who caught many issues and pushed us to improve our paper on many fronts.

Our collaboration on this project first started at the Perimeter Institute, 
and we thank it for the wonderfully supportive, convivial atmosphere that fosters adventurous work.
The National Science Foundation supported O.G. through DMS Grants No. 1812049 and 2042052, which allowed K.R. to visit and push this project forward.

\section{Local functionals}
\label{sec:local}


Let $X$ be a smooth manifold of dimension $n$, and consider a field theory living on~$X$.

Our overarching goal is to introduce a version of multilocal functionals --- a natural class of observables for a field theory --- that is a cosheaf on $X$ with respect to the {\em Weiss topology}, 
which is the challenging condition necessary to obtain a {\em factorization} algebra of observables.
This cosheaf condition will be obtained by using a partitions of unity argument,
modeled on the standard example of compactly-supported sections of a vector bundle.

We introduce, in fact, two versions of (multi)local functionals. 
The first is the conventional notion, as used in \cite{FR,FR3},
and we show it satisfies a cosheaf condition in a non-homotopical sense.
The second replays the tune in a derived mode --- it will be familiar to those fond of the variational bicomplex --- that provides a homotopy cosheaf. 
In this section we treat local functionals, and discuss several variations and useful conditions.

For a discussion of cosheaves and homotopy cosheaves,
see Appendix~A.5 of~\cite{CoGw},
which is motivated by similar considerations.
It also contains guidance to useful pedagogical literature.

\subsection{Warm-up example of a cosheaf}

If $E \to X$ is a (graded) vector bundle,
let $\Ecal$ denote the sheaf on $X$ of smooth sections of $E$.
Every section can be multiplied by a smooth function,
and so $\Ecal$ is a sheaf of modules over the sheaf $\cinfty$ of smooth functions.
We thus have access to partitions of unity, 
which lead to a standard and straightforward proof of the following.

\begin{lemma}
The functor $\Ecal_c: \Opens(X) \to \Vec$, where $\Ecal_c(U)$ is the space of compactly supported smooth sections of $E$, is a cosheaf on~$X$.
\end{lemma}

To see that $\Ecal_c$ is a functor (i.e., a precosheaf), 
observe that a compactly-supported section $\phi \in \Ecal_c(U)$ extends to a compactly-supported section $\phi \in \Ecal_c(V)$ of a larger open $V \supset U$ by setting $\phi = 0$ on $V -U$.

To see it is a cosheaf, the key idea is that partitions of unity allow one to write any section $f \in \Ecal_c(U)$ as a linear combination of such sections subordinate to the cover.
In other words, the sections are constructed in a local-to-global fashion.
See \S5.4 of Appendix A in \cite{CoGw} for a proof, 
but other proofs can be found in~\cite{BottTu, Bredon}.

\begin{rem}
Consider the more interesting case where we have a differential complex $(\Ecal_c, Q)$,
where $Q$ is a differential operator that defines a cochain complex.
(As an example, think of the de Rham complex.)
Then $(\Ecal_c,Q)$ satisfies another local-to-global condition:
it is a homotopy cosheaf, meaning that $(\Ecal_c(U),Q)$ is quasi-isomorphic to the homotopy colimit over the full simplicial diagram arising from a cover.
This proof can also be seen in \S5.4 of Appendix A in~\cite{CoGw}.
\end{rem}

\subsection{A sketch of why observables form a precosheaf}
\label{sec on obs as precosheaf}

Let $E \to X$ be a vector bundle on a smooth manifold $X$,
and let $\Ecal$ denote the sheaf of smooth sections.
We view $\Ecal$ as describing the fields of a field theory,
although for the moment we will discuss only kinematical structures and hence ignore the interesting aspects of field theory.

On an open set $U$, the vector space $\Ecal(U)$ can be viewed as an infinite-dimensional Fr\'echet manifold.
Let $\cinfty(\Ecal(U))$ denote the commutative algebra of smooth functions on this manifold.
(To be precise, we will think about Bastiani-smooth functionals, 
but here these coincide with the convenient notion, since we are working with Fr\'echet spaces. \cite{Michor})
We view these as the observables for the field theory that only depend on the behavior of fields in the region~$U$.

Since $\Ecal$ is a sheaf and hence is a contravariant functor out of the poset category $\Opens(X)$, we see that we have a precosheaf
\[
\begin{array}{cccc}
\cinfty(\Ecal):& \Opens(X) & \to & {\rm CAlg}\\
& U & \mapsto & \cinfty(\Ecal(U))
\end{array}
\]
because taking functions is a contravariant functor on the category of Fr\'echet manifolds.
We view this precosheaf as describing observables for the fields encoded by~$\Ecal$.

In the remainder of this paper we will develop variations on this construction,
notably by working with replacements of smooth functions (often subalgebras)
motivated by constructions in physics.
We will always produce precosheaves, and we will want to understand when these constructions are cosheaves.

There is a special feature of this linear situation.
The sheaf $\Ecal$ is {\em soft}, so that any restriction map $\Ecal(X) \to \Ecal(K)$ for {\em closed} set $K \supset X$ is surjective.
Hence the map $\cinfty(\Ecal(K)) \to \cinfty(\Ecal(X))$ is {\em injective}.

\begin{df}
\label{df_supp}
For an observable $F \in \cinfty(\Ecal(X))$, its {\em support} $\supp(F)$ is the smallest closed set $K$ such that $F$ is in the image of the map $\cinfty(\Ecal(K)) \to \cinfty(\Ecal(X))$.
\end{df}

For example, a constant function (i.e., observable that assigns the same number to every field) has support given by the empty set.
As another example, consider the case where $E$ is the trivial line bundle so that a field $\phi$ is simply a smooth function on $X$, and suppose $X$ is an oriented $n$-dimensional manifold.
Then the observable
\[
F(\phi) = \int_X \phi(x) f(x) \d^n x,
\]
with $f(x) \d^n x$ denoting some compactly supported top form,
has support given by the support of the top form.

\subsection{Local functionals, the strict version}

We turn now to a more sophisticated situation.
Let $\loc(X)$ denote the space of {\em local} functionals on the field configurations $\Ecal(X)$.
In brief, these are functions on $\Ecal(X)$ that arise by integrating Lagrangian densities.
For $E \to X = \RR^n$ the trivial line bundle, so that a field $\phi$ is simply a smooth function on $\RR^n$, 
a concrete example of a local functional is
\begin{equation}
\label{eq: rep of local}
F(\phi) = \int_{\RR^n} \phi(x)^3 (\partial_1 \phi(x))^7 (\partial_1 \partial_n^2 \phi(x)) f(x)\, \d^n x
\end{equation}
where $f$ is a compactly supported smooth function.
(The compact support of $f$ ensures that this formula makes sense for arbitrary fields without any decay at ``infinity.'')
Typically, a local functional involves integrating some polynomial in the derivatives of the field~$\phi$.

We will extend this characterization to general vector bundles over smooth manifolds.
We will follow chapter 3 of \cite{Book}, chapter 2 of Deligne and Freed in~\cite{FieldsStrings}, 
and \cite{BDGR},
where more discussion and motivation can be found.

To capture the data of the field and its derivatives, 
we use the jet bundles $\pi^{(k)}: J^k E \to X$,
where $k$ runs over the natural numbers. 
Recall that for any section $\phi \in \Ecal$, 
there is a canonical section $\jet_k \phi$ of $J^k E$, 
known as the {\it $k$th jet prolongation} of $\phi$, 
such that $\jet_k \phi(x)$ records the Taylor expansion of $\phi$ to order $k$ at each point of $x \in~X$.

Let $\Dens_X \to X$ denote the density line bundle (or orientation line bundle) of $X$.
When $X$ is oriented, $\Dens_X$ is precisely the top exterior power of the cotangent bundle.
Let $\cDens$ denote the smooth sections of this bundle,
which forms a sheaf.

We can now define abstractly the kind of integrands that appear in~\eqref{eq: rep of local}.
There are two parts to such a local functional: the Lagrangian and the density.

\begin{df}
\label{order k lag}
A {\em Lagrangian} of order $k \in \NN$ is a smooth function on the total space ${\rm J}^k E$ of the vector bundle $J^k E \to X$.
\end{df}

Given a field $\phi$, its prolongation $\jet_k \phi$ is a section of $J^k E$,
and so it composes with an order $k$ Lagrangian $\alpha$ to define a function $\alpha(\jet_k \phi)$ on $X$.
We can pair it with a compactly-supported density $\mu$ on $X$ to obtain a functional
\[
F_{\alpha,\mu}(\phi) = \int_X \alpha(\jet_k \phi) \mu
\]
on fields. 
This functional is the kind of object we are interested in.

Note an issue that will recur repeatedly.
We require $\mu$ to be compactly-supported so that we can evaluate $F_{\alpha,\mu}$ on an arbitrary field $\phi \in \Ecal(X)$, which may be nonzero everywhere.
If we wish to allow arbitrary densities $\mu$ (i.e., without compact support), 
the formula for $F_{\alpha,\mu}$ is well-defined if we restrict to compactly-support fields $\phi \in \Ecal_c(X)$.
This situation appears, for example, with action functionals, 
which are usually defined everywhere on spacetime and do not define functionals on arbitrary fields.
Depending on our needs, we simply place the compact support condition in one place or the other.
For the moment, we will allow arbitrary densities.

\begin{df}
Let $\loc_{(k)}(X)$ denote the subspace of $C^\infty(\Ecal_c(X))$ spanned by functionals $F_{\alpha,\mu}$ where
$\alpha$ is an order $k$ Lagrangian and $\mu$ is a smooth density.
\end{df}

We now turn to examining how this construction works functorially in open subsets of the manifold.
Note that if $U \subset V$ is an inclusion of opens in~$X$,
there is a restriction map $\loc_{(k)}(V) \to \loc_{(k)}(U)$ .
This feature is natural: 
if a field $\phi$ has compact support in $U$, 
then one can extend $\phi$ to a larger open $V \supset U$,
and hence a functional $F$ defined on $\Ecal_c(V)$ can be evaluated on $\phi$, so $F$ determines a functional on $\Ecal_c(U)$.
By varying over opens in $X$, the construction $\loc$ thus defines a functor from $\Opens(X)^{op}$ to graded vector spaces.
In other words, $\loc_{(k)}$ is a {\em sheaf}.

\begin{df}
\label{def of loc}
Let $\loc$ denote the colimit over $k$ of the sheaves~$\loc_{(k)}$.
\end{df}

\begin{rem}
Note that a global section $F \in \loc(X)$ need not have a finite order $k$, i.e., may not be an element of some $\loc_{(k)}(X)$.
A colimit of sheaves is the sheafification of the colimit of presheaves,
so a global section $F$ of $\loc(X)$ has the property that for every point $x \in X$, there is some neighborhood $U$ where $F|_U$ has finite order but this order might vary in the choice of point $x$.
Thus, unless $X$ is compact, a global bound on order is not guaranteed.
This kind of issue will recur repeatedly in the paper.
\end{rem}



\begin{rem}
The notation Loc is often used for a category of spacetimes in the AQFT community.
We do not work in that setting in this paper,
so we use it for local functionals.
\end{rem}

Note, however, that different choices of pairs $(\alpha,\mu)$ can produce the same functional.
In particular, note that $(f \alpha, \mu)$ and $(\alpha, f \mu)$ determine the same functional when $f \in \cinfty(X)$
because they determine the same integrand 
\[
(f \alpha(\jet_k \phi)) \mu = \alpha(\jet_k \phi) (f\mu)
\]
for every field $\phi$, i.e., these integrands agree pointwise, even before integrating.
We now describe the vector space of such integrands in more abstract terms.

Let $\pi^k$ denote the projection map from $J^k E$ to~$X$.
By pulling back along $\pi^k$, every function $f \in \cinfty(X)$ becomes a function on ${\rm J}^k E$,
and so the order $k$ Lagrangians $\cinfty({\rm J}^k E)$ are a $\cinfty(X)$-module.
The densities $\cDens(X)$ are also a $\cinfty(X)$-module.

\begin{df}
The order $k$ {\em Lagrangian densities} form a sheaf assigning to an open set $U$, the vector space
\[
\lag_{(k)}(U) = \cinfty({\rm J}^k E|_U) \otimes_{\cinfty(U)} \cDens(U),
\] 
and the {\em Lagrangian densities} are the colimit over $k$ of  the sheaves~$\lag_{(k)}$.
\end{df}
 
There is a canonical map
\[
q_{(k)}: \lag_{(k)} \to \loc_{(k)}
\]
sending $\alpha \otimes \mu$ to~$F_{\alpha,\mu}$.
Taking the colimit over $k$, we obtain a map
\[
q: \lag(X) \to \loc(X).
\]
In the next subsection we discuss how total derivatives are in the kernel of~$q$.

\begin{rem}
\label{analyticity of loc}
We have allowed the Lagrangian $\alpha$ to be a smooth function on a jet bundle.
In practice the primary examples of interest for physicists are polynomial or analytic functions {\it along the fiber} of the jet bundle:
consider, for instance, how most action functionals involve polynomials built out of derivatives of the fields.
(For a more detailed discussion of this notion, see Section~\ref{analytic variant} below.)
The restriction to fiberwise analytic functions becomes essential when we quantize, 
as Epstein-Glaser renormalization uses power series in a central way.
In defining local and multilocal functionals --- and classical observables --- we will work with smooth Lagrangians.
But in section~\ref{analyticity of obscl} we make the restriction to fiberwise analytic functionals in order to quantize.
\end{rem}

In practice, we want functionals on $\Ecal$, the noncompactly supported fields, not on the compactly supported fields.
Thus, let $\loc_{(k),c}(X)$ denote the subspace of $C^\infty(\Ecal(X))$ spanned by functionals $F_{\alpha,\mu}$ where $\alpha$ is a Lagrangian of order $k$ and $\mu$ is a compactly-supported smooth density.
This notion is also functorial in opens, but it is covariant,
so $\loc_{(k),c}$ forms a {\em precosheaf} on $X$.
Let $\loc_c$ denote the precosheaf of compactly supported local functionals.

We also have a precosheaf $\lag_c$ of compactly supported Lagrangian densities. 
Integration again determines a map $q: \lag_c \to \loc_c$ that produces a local functional from a compactly-supported Lagrangian density.

We want to show that $\loc_c$ is, in fact, a cosheaf.
As a first step, notice the following.

\begin{lemma}
The functor $\lag_c$ is a cosheaf.
\end{lemma}

\begin{proof}
To prove this, we need to borrow from our warm-up example.
As densities are sections of a vector bundle, 
one can take compactly supported sections to get a cosheaf, 
thanks to partitions of unity.
One immediately extends this observation to compactly supported Lagrangian densities as follows. 
For any $\phi \in \Ecal$ and any compactly supported Lagrangian density $\alpha \in \lag_c$, 
we have a density $\alpha(\jet_k \phi)$. 
A partition of unity $\{\psi_j\}_{j \in J}$ lets this density be decomposed subordinate to an arbitrary cover as
\[
\alpha(\jet_k \phi) = \sum_j \psi_j \alpha(\jet_k \phi)
\]
This construction is independent of $\phi$ because $\alpha = \sum_j \psi_j \alpha$ decomposes $\alpha$ subordinate to the cover as well.
\end{proof}

We now state a claim of import to us.

\begin{lemma}
\label{lem: loc is cosheaf}
This functor $\loc_c$ is a cosheaf on $X$ with values in graded vector spaces. 
\end{lemma}

For the sake of clarity, we mean here the topology on $X$ as a topological space, 
not the Grothendieck site on the category $\Opens(X)$ determined by Weiss covers,
which appears when we discuss factorization algebras.

In the proof, we will use some facts about local differential forms, 
which are defined in the next section.
Hence the proof appears after remark~\ref{proof of loc is cosheaf} in that section.

\subsection{Local functionals, the dg version}

We would like to have a {\em homotopy} cosheaf of local functionals,
i.e., we would like $\loc_c(U)$ to be quasi-isomorphic to the homotopy colimit of $\loc_c$ on the simplicial diagram arising from any sufficiently nice cover $\{U_i\}$ of an open subset~$U$.
(
Unfortunately, it does not seem that $\loc_c$ is a homotopy cosheaf,
so we replace our notion of local functional above by a derived version that does have this local-to-global behavior.
Loosely speaking, we replace densities $\cDens$ by the de Rham complex;
instead of working with Lagrangian densities modulo total derivatives,
this complex itself implements that quotient when taking cohomology.
Moreover, we can then use partitions of unity acting on the de Rham complex to check the homotopy cosheaf condition.

This kind of issue is analogous to the use of resolutions in sheaf theory.
For example, the constant sheaf $\underline{\RR}$ on $X$ is a sheaf of vector spaces,
but it is not a homotopy sheaf (i.e., a sheaf with values in the $\infty$-category of cochain complexes).
The de Rham complex provides a soft resolution of the constant sheaf $\underline{\RR}$ and it does define a homotopy sheaf.x

\begin{rem}
For those familiar with the variational bicomplex, 
we will use a quotient complex that consists of local functions on the fields with values in de Rham forms on the spacetime.
In other words, it is the bottom horizontal row of the variational bicomplex.
\end{rem}

We now explain our construction in detail.

\def\orlb{{\rm or}}

First, we explain the de Rham complex associated to densities.
When $X$ is an oriented $n$-manifold, the density line $\Dens \to X$ is precisely $\Lambda^n T^*_X$,
so that densities $\cDens$ are the smooth top forms.
{\it We will assume throughout this paper that $X$ is oriented,}
as it is a standard hypothesis in algebraic quantum field theory and because it is straightforward (but notationally distracting) to deal with the general case.
(See the next remark.)

Thus, we see 
\[
\Omega^0(X) \xto{\d} \Omega^1(X) \xto{\d} \cdots \xto{\d} \Omega^n(X) = \cDens(X).
\]
If we shift this complex and place it in degrees $-n$ to $0$, 
then densities sit in degree 0,
and the zeroth cohomology consists of densities modulo total derivatives.
As we are working with sections of vector bundles and the differential in this complex is a differential operator,
it manifestly determines a sheaf on~$X$.

\begin{rem}
When $X$ is unoriented, however, the density line is the tensor product $\Lambda^n T^*_X \otimes \orlb_X$,
where $\orlb_X \to X$ denotes the orientation line bundle,
which has a natural flat connection.
We can use instead the de Rham complex for this flat line bundle.
Note that $X$ is oriented if and only if the orientation bundle is trivial.
\end{rem}

Second, we now provide a new, cochain version of Lagrangian densities.
The idea is just a simple extension of the earlier version: 
given a function $\alpha \in C^\infty({\rm J}^k E)$ and a $p$-form $\mu$ on $X$,
we consider the {\em Lagrangian $p$-form}
\[
\alpha(\jet_k(\phi)) \mu,
\]
where $\phi$ denotes an arbitrary field.
Such a Lagrangian $p$-form is a map from $\Ecal(X)$ to $p$-forms on $X$;
in fact, by construction, 
\[
\alpha(\jet_k(-)) \mu: \Ecal \to \Omega^p_X
\]
is  a map of sheaves on $X$.
In this precise sense it is local on~$X$.

Let us make a more explicit characterization.
Recall that $\pi^k$ denotes the projection map from ${\rm J}^k E$ to~$X$,
and let $\pi^k_* C^\infty_{{\rm J}^k E}$ denote the sheaf obtained by pushing forward along $\pi^k$ of smooth functions on the total space of this jet bundle.
(That means this sheaf on $X$ satisfies
\[
\pi^k_* C^\infty_{{\rm J}^k E}(U) = C^\infty_{{\rm J}^k E}((\pi^k)^{-1}(U))
\]
for any open set $U \subset X$.)
An $\alpha$ as above can be viewed as a section of this sheaf on $X$,
and the composite $\alpha \circ \jet_k$ is a map of sheaves from $\Ecal$ to $C^\infty_X$.
(Note that we mean here maps of sets, not linear maps.)

\begin{lemma}
The sheaf map
\[
-\circ \jet_k: \pi^k_* C^\infty_{{\rm J}^k E} \to {\rm Maps}(\Ecal, C^\infty_X).
\]
is injective on each open $U \subset X$.
\end{lemma}

\begin{proof}
The map is linear, by inspection, so it suffices to show that the kernel is zero.
Suppose $\alpha$ in $\pi^k_* C^\infty_{{\rm J}^k E}$ is sent to zero in ${\rm Maps}(\Ecal, C^\infty_X)$.
Then for every field $\phi$, we have that $\alpha(\jet_k(\phi)) = 0$, 
but this implies $\alpha$ vanishes at every point $x$ in ${\rm J}^k E$,
as we can produce a $\phi$ whose prolongation passes through~$x$.
Hence $\alpha = 0$, as desired.
\end{proof}

This set-up carries over to maps with values in $p$-forms as well.
Consider the composite map
\[
\pi^k_* C^\infty_{{\rm J}^k E} \times \Omega^p_X \xto{(-\circ \jet_k) \times \id } {\rm Maps}(\Ecal, C^\infty_X) \times \Omega^p_X \xto{\cdot} {\rm Maps}(\Ecal, \Omega^p_X)
\]
sending $(\alpha, \mu)$ to the functional $\phi \mapsto \alpha(\jet_k(\phi)) \mu$.
This map is bilinear by inspection, so we can work with the (algebraic) tensor product $\pi^k_* C^\infty_{{\rm J}^k E} \otimes \Omega^p_X$. 
Note, however, that for $f$ a smooth function on $X$, 
the image of $(f\alpha, \mu)$ and $(\alpha, f\mu)$ are the same.
We thus want to work with
\begin{equation}
\label{eqn Fp}
\Upsilon^p: \pi^k_* C^\infty_{{\rm J}^k E} \otimes_{\cinfty_X} \Omega^p_X \to {\rm Maps}(\Ecal, \Omega^p_X)
\end{equation}
and we call any functional of the form $\Upsilon^p(\alpha, \mu)$ a {\em Lagrangian $p$-form}.

\begin{rem}
Note that $\pi^k_* C^\infty_{{\rm J}^k E}$ is a subsheaf of ${\rm Maps}(\Ecal, C^\infty_X)$,
but it is closed under the right action of $C^\infty_X$ via the canonical map $(\pi^k)^*: C^\infty_X \to \pi^k_*C^\infty_{{\rm J}^k E}$ of sheaves of algebras.
Hence we can take the relative tensor product used in the definition.
\end{rem}

\begin{df}
\label{lag p-forms}
The sheaf of order $k$ {\em Lagrangian $p$-forms} denotes 
\[
\llag_{(k)}^p = \pi^k_* C^\infty_{{\rm J}^k E} \otimes_{C^\infty_X} \Omega^p_X,
\]
which embeds via $\Upsilon^{p}$ into~${\rm Maps}(\Ecal, \Omega^p_X)$.

Let the {\em Lagrangian $p$-forms} $\llag^p(X)$ denote the colimit over $k$ of the~$\llag_{(k)}^p(X)$.
\end{df}

Having defined the Lagrangian $p$-forms, we need to explain why there is a de Rham-type operator going from $p$-forms to $p+1$-forms and yielding a cochain complex.

Observe that for any smooth field $\phi$,
the $p$-form $\alpha(\jet_k(\phi)) \mu$ is smooth because the prolongation $\jet_k$ and $\alpha$ are smooth.
Thus it is possible to apply the usual exterior derivative $\d$ to this $p$-form.

\begin{lemma}
\label{lem: d makes sense}
For any Lagrangian $p$-form $\Lcal$, there is a Lagrangian $p+1$-form $\Lcal'$ such that $\d(\Lcal(\phi)) = \Lcal'(\phi)$ for any field~$\phi$.
\end{lemma}

\begin{proof}
By applying a partition of unity, we can check the claim locally in a coordinate patch over which we have chosen a frame to trivialize the bundle $E$.
Moreover, it suffices to check on a Lagrangian $p$-form of the form $\alpha(\jet_k(-)) \mu$, as above.
For an arbitrary field $\phi$, we then compute
\[
\d\left( \alpha(\jet_k(\phi)) \mu\right) =\d\left( \alpha(\jet_k(\phi)) \right)\mu + \alpha(\jet_k(\phi))  \d \mu.
\]
The rightmost term is manifestly a Lagrangian $p+1$-form, so we just need to check that the first summand is too.

Let $(x_1,\ldots, x_n)$ be coordinates in the patch, and with respect to the frame, a field decomposes as $\phi = (\phi_1,\ldots,\phi_m)$.
These choices lead to a natural frame on the jet bundle ${\rm J}^k E$ given by the $(\partial^\mu \phi_i)$ with multi-index $\mu = (m_1,\ldots,m_n)$ such that $|\mu| = \sum_{i=1}^n m_i \leq k$;
this frame plus the base coordinates equip the total space of the jet bundle with a coordinate system.

The chain rule tells us
\begin{equation}
\label{defofdlag}
\d\left( \alpha(\jet_k(\phi)) \right) 
= \sum_{i = 1}^n \left( 
\frac{\partial \alpha}{\partial x_i}\Big|_{\jet_k(\phi)(x)} 
+ \sum_{j, |\mu| \leq k} \frac{\partial \alpha}{\partial (\partial^\mu \phi_j)}\Big|_{\jet_k(\phi)(x)} \frac{\partial (\partial^\mu \phi_j )}{\partial x_i}
\right) \d x_i.
\end{equation}
A term like $\partial \alpha/\partial x_i$ manifestly encodes a function on the jet bundle ${\rm J}^k E$.
A term like $\partial \alpha/\partial (\partial^\mu \phi_j)$ is also a function on the jet bundle,
so the second term arises by pulling back the function 
\[
\frac{\partial \alpha}{\partial (\partial^\mu \phi_j)} \frac{\partial}{\partial x_i}  \left(\partial^\mu \phi_j \right)
\]
from the jet bundle ${\rm J}^{k+1} E$ along the prolongation $\jet_{k+1} \phi$.
(Note that this function might depend on a derivative of a $k$-jet of the field.)
\end{proof}

Let $\d_{\llag}$ denote the operator on Lagrangian differential forms that is induced by the exterior derivative $\d$ on differential forms:
if $F_{\alpha,\mu}$ denotes the Lagrangian $p$-form associated to $\alpha$ and $\mu$, we define
\[
(\d_{\llag} F_{\alpha,\mu})(\phi) = \d(\alpha(\jet_k(\phi))\mu). 
\]
An explicit description in coordinates is given in equation~\eqref{defofdlag} above.
This induced operator inherits the usual properties of the exterior derivative,
and, in particular, is square-zero.

This observation suggests that we work with the cochain complex~$\llag^\bullet$, given by
\begin{equation}
\label{olver complex}
\llag^0 \xto{\d_\llag} \llag^1 \xto{\d_\llag} \cdots \xto{\d_\llag} \llag^n   
\end{equation}
and concentrated in degrees $0$ to n,
where $\d_{\llag}$ denotes the differential induced by the exterior derivative and constructed above in Lemma~\ref{lem: d makes sense}.
The cohomology in degree $n$ contains, as a summand, the local functionals already defined.
There is, however, some unwanted cohomology, even locally, by the following result.

\begin{thm}[\cite{Olver}, Theorem~5.80]
\label{olver thm}
On any star-convex open $U \subset \RR^n$, 
the cohomology of the complex $\llag^\bullet(U)$ is $\RR$ in degree~$0$, $\loc(U)$ in degree~$n$, and vanishes in all other degrees.
\end{thm}

The cohomology in degree~$0$ arises from the fact that a Lagrangian that is constant everywhere is annihilated by $\d_{\llag}$.
This theorem is a sophisticated cousin of the Poincar\'e lemma. See Section 5.4 of \cite{Olver} for a systematic treatment. 
(It is also sometimes called the {\em algebraic Poincar\'e lemma}; see \cite{H}[Section 2.3].)

There is a simple way to excise this unwanted cohomology.
There is a canonical cochain map $\iota$ from the de Rham complex $\Omega^\bullet_X$ to the complex~\eqref{olver complex} sending a $p$-form to itself, viewed as a Lagrangian $p$-form that does not depend on a field $\phi$.
Take the mapping cone of this inclusion map,
which is our key object.

\begin{df}\label{df:dL}
Let $\lloc^\bullet$ denote the cochain complex of sheaves
\[
\left( \llag^\bullet[n] \oplus \Omega^\bullet_X[n+1], \d_\lloc \right)
\]
concentrated in degrees $-n-1$ to 0,
where $\d_\lloc$ denotes the sum $\d_\llag \pm \d_{DR} + \iota$ 
with $\d_\llag$ the differential constructed in Lemma~\ref{lem: d makes sense}, with $\pm \d_{DR}$ the exterior derivative acting on the shifted de Rham complex, and with $\iota$ viewed now as a degree~1 map between the summands.

The compactly-supported sections $\lloc^\bullet_c$ of this complex determine the {\em local differential forms} on the fields~$\Ecal$.
\end{df} 

We have shifted the complexes so that Lagrangian densities appear in cohomological degree zero,
as they are the objects of central interest to us.
Indeed, we introduce this complex to find a soft resolution of local functionals.
In explicit terms
\[
\lloc^k = 
\begin{cases} 
\Omega^0_X & k = -n-1 \\
\Omega^{n+k+1}_X \oplus \llag^{n+k} & -n \leq k \leq -1 \\
\llag^n & k = 0 \\
0 & \text{else}
\end{cases}.
\]
Note that, by definition, $\llag^n$ is isomorphic to $\lag$ 
as sheaves.

There is thus a map $\int$ of sheaves of cochain complexes by the composite
\[
\lloc_c^\bullet \xto{\tau_{\geq 0}} \llag_c^n = \lag_c \xto{q} \loc_c
\]
where $\tau_{\geq 0}$ is the truncation to degree zero, which keeps only the top form part of an element of $\lloc^\bullet$.
This composite map extracts a traditional local functional from a local differential form.

As a corollary of Theorem~\ref{olver thm} and the properties of a mapping cone,
we have the following.

\begin{lemma}
On any star-convex open $U \subset \RR^n$, the map $\int_U: \lloc_c^\bullet(U) \to \loc_c(U)$ is a quasi-isomorphism.
\end{lemma}

This lemma implies locally --- when $U$ is a ball --- there is no essential difference between the dg version of local functionals and the traditional version. 
On the other hand, for an open set $U$ with interesting topology (e.g., with nontrivial higher de Rham cohomology), 
the map $\int_U$ may {\em not} be a quasi-isomorphism.
Indeed, there may be some Lagrangian $p$-form that pairs with a nontrivial $p$-cycle in $U$ to produce an interesting observable. 
Extended operators, in the style of Wilson loops, provide examples.

\begin{prop}
\label{lem: lloc is homotopy cosheaf}
Consider the functor $\lloc^\bullet_c: \Opens(X) \to \Ch$ that assigns to an open subset $U$ of $X$, 
the cochain complex $\lloc^\bullet_c(U)$ of local differential forms with compact support in $U$.
It is a homotopy cosheaf on $X$ with values in cochain complexes (with quasi-isomorphisms as weak equivalences).
\end{prop}

Before proving this result, 
we remark on an important consequence,
which is analogous to facts about the sheaf of de Rham complexes and the constant sheaf.
Recall that the de Rham complex is quasi-isomorphic to the constant sheaf on star-convex opens, but it is not quasi-isomorphic on non-contractible manifolds.
Thus, while the constant sheaf $\underline{\RR}$ is a strict sheaf, 
it is not a homotopy sheaf;
its {\em derived} global sections is the cohomology $H^*(X,\RR)$,
which can be computed via the de Rham complex.

Similarly, a (homotopy) cosheaf is determined by its behavior on a base for the topology.
As the  $\loc_c$ and $\lloc_c$ are quasi-isomorphic locally (i.e., on the base of star-convex opens),
they determine the same homotopy cosheaf.
We have seen that $\lloc_c$ gives an explicit description of this homotopy cosheaf on ``big'' opens (e.g., on opens that are not disjoint unions of star-convex opens),
but $\loc_c$ and $\lloc_c$ are {\em not} quasi-isomorphic on all opens.
Hence, we see the following.

\begin{cor}
The cosheaf $\loc_c$ is not a homotopy cosheaf.
\end{cor}

\begin{rem}
On the other hand, $\lloc_c$ is a strict cosheaf with values in cochain complexes, because it is constructed as compactly-supported sections of soft sheaves.
Again, compare with the de Rham complex.
\end{rem}

\begin{proof}[Proof of Proposition~\ref{lem: lloc is homotopy cosheaf}]
The proof is inspired by the standard argument to show that \v{C}ech and de Rham cohomology are isomorphic on a manifold.
In that argument, one picks a good cover of an open $U$ and then makes a double complex by taking the \v{C}ech complex of the de Rham complex.
In the spectral sequence for the double complex, if one takes the \v{C}ech differential first,
one finds that the next page of the spectral sequence is the de Rham complex of $U$,
which is a double complex concentrated in a single column (or row, depending on one's convention),
so that the sequence collapses on the next page.
This argument shows that the \v{C}ech cohomology of the de Rham complex agrees with the de Rham cohomology.

Consider now our situation. 
Pick a cover of the open $U$.
The homotopy colimit of $\lloc^\bullet_c$ over the \v{C}ech nerve of this cover can be computed as the  \v{C}ech complex for $\lloc^\bullet_c$.
(See \S C.5 of \cite{CoGw} for a proof and extended discussion in precisely this context, 
but it is a standard fact that homotopy colimits of simplicial diagrams of cochain complexes are modeled by totalizing well-known double complexes.)
The \v{C}ech complex for $\lloc^\bullet$ is, in fact, the totalization of a double complex, 
where one direction arises from the cover (lets call it the \v{C}ech direction) and the other direction arises from the differential~$\d_\lloc$.
We will then use a variant of the \v{C}ech-de Rham argument to show we have a homotopy cosheaf.

If we focus on the \v{C}ech direction for some fixed layer $\lloc_c^p$ of $\lloc_c^\bullet$,
we see the \v{C}ech complex for compactly supported sections of a vector bundle.
A partition of unity lets one construct a contracting homotopy for this \v{C}ech complex,
showing that $\lloc_c^p$ is a homotopy cosheaf.
(See \S A.5.4 of \cite{CoGw} for a detailed proof in  this setting.)

Consider now the spectral sequence for a double complex, encoding the \v{C}ech complex of all of $\lloc_c^\bullet$, that uses the \v{C}ech differential first.
We thus see that the next page of the sequence consists precisely of the ``naive'' global sections~$\lloc_c^\bullet(U)$.
It corresponds to a double complex concentrated along a single column (i.e., sitting in one degree of the \v{C}ech direction) with only the differential of $\lloc_c^\bullet(U)$.
Thus the spectral sequence collapses on the next page.
\end{proof}

We now have the ingredients to prove that $\loc_c$ is a cosheaf with values in graded vector spaces.

\begin{rem}
\label{proof of loc is cosheaf}
The crucial new ingredient is that $\loc_c$ is the cokernel of a map of cosheaves
\[
\llag_c^{n-1} \xto{\d_\llag} \llag_c^n
\] 
because two Lagrangian densities determine the same local functional precisely when they differ by a total differential.
A colimit of cosheaves agrees with the colimit of precosheaves (just as a limit of sheaves can be computed as a limit of presheaves), 
so $\loc_c$ is a cosheaf.
In addition to this high-brow argument, we will now give an explicit argument.
\end{rem}

\def\ext{{\rm ext}}

\begin{proof}[Proof of lemma~\ref{lem: loc is cosheaf}]
Let $U$ be an open set and let $\{U_i\}_{i \in I}$ be a cover of $U$.
Let $\ext_{i}: \loc_c(U_i) \to \loc_c(U)$ denote the extension map for $U_i \subset U$, and let $\ext_{ij,j}: \loc_c(U_i\cap U_j) \to \loc_c(U_j)$ denote the the extension map for $U_i \cap U_j \subset U_j$.
In concrete terms,
given $\phi \in \Ecal(U)$ and $F \in \loc_c(U_i)$, set
\[
(\ext_i F)(\phi) = F\left(\phi\big|_{U_i}\right).
\]
We need to show that $\loc_c(U)$ is the cokernel of the map
\[
\epsilon: \bigoplus_{j,k} \loc_c(U_j \cap U_k) \to \bigoplus_i \loc_c(U_i) 
\]
that sends an element $F_{jk} \in \loc_c(U_j \cap U_k)$ to $(\ext_{jk,j} F_{jk}, - \ext_{jk,k} F_{jk}) \in \loc_c(U_j) \oplus \loc_c(U_k)$.
Any element $F$ in the domain is a linear combination of such~$F_{ij}$.

Note that $\loc_c(U)$ receives a canonical map from $\bigoplus_i \loc_c(U_i)$: send $(F_i)_{i \in I} \in \bigoplus_i \loc_c(U_i)$ to the sum of the extensions $\sum_i \ext_i F_i$
This canonical map 
\[
c: \bigoplus_i \loc_c(U_i) \to \loc_c(U)
\]
factors through ${\rm coker}\,\epsilon$ because the composite $c \circ \epsilon$ is quickly seen to be zero.
Let $\bar{c}: {\rm coker}\,\epsilon \to \loc_c(U)$ denote the induced map.
We need to show $\bar{c}$ is an isomorphism.

It is helpful to keep in mind the relationship with Lagrangian differential forms.
Consider the commuting diagram
\[
\begin{tikzcd}
\bigoplus_{j,k} \llag^{n-1}_c(U_j \cap U_k) \ar[r, "\epsilon^{n-1}"] \ar[d, "\d_\llag"] & \bigoplus_i \llag^{n-1}_c(U_i) \ar[r, "c^{n-1}"] \ar[d, "\d_\llag"] & \llag^{n-1}_c(U) \ar[r] \ar[d, "\d_\llag"] & 0 \\
\bigoplus_{j,k} \llag^n_c(U_j \cap U_k) \ar[r, "\epsilon^n"] \ar[d, "\int"] & \bigoplus_i \llag^n_c(U_i) \ar[r, "c^n"] \ar[d, "\int"] & \llag^n_c(U) \ar[r] \ar[d, "\int"] & 0 \\
\bigoplus_{j,k} \loc_c(U_j \cap U_k) \ar[r, "\epsilon"] \ar[d] & \bigoplus_i \loc_c(U_i) \ar[r, "c"] \ar[d] & \loc_c(U) \ar[r] \ar[d] & 0 \\
0 & 0 & 0 & 
\end{tikzcd}
\]
Each column is exact by the definition of local functionals. 
The rows for $\llag^n$ and $\llag^{n-1}$ are exact as they are cosheaves.
(The horizontal arrows are extensions maps; 
they are similar to the maps $c$ and $\epsilon$ that we just defined.)
We wish to show the bottom row is exact as well,
which is equivalent to showing $\bar{c}$ is an isomorphism.

To start, we wish to show that $\bar{c}$ is surjective so we show $c$
\[
\bigoplus_i \loc_c(U_i) \to \loc_c(U)
\]
is surjective.
Given $F \in \loc_c(U)$,
it can be expressed as a finite linear combination 
\[
F = F_{\alpha_1, \mu_1} + \cdots + F_{\alpha_n, \mu_n}
\]
where the $\mu_i$ have compact support $K_i \subset U$.
Pick a partition of unity $\{ \psi_i \}_{i \in I}$ subordinate to the cover. 
Then $F_{\alpha_m, \psi_i \mu_m}$ is an element of $\loc_c(U_i)$.
Note that
\begin{align*}
\sum_{i \in I} F_{\alpha_m, \psi_i \mu_m} (\phi) 
&= \sum_{i \in I} \int_{U} \alpha_m(\jet_k \phi) \psi_i \mu_m \\
&= \int_U \alpha_m(\jet_k \phi) \left(\sum_{i \in I} \psi_i \right) \mu_m \\
&= \int_U \alpha_m(\jet_k \phi) \mu_m \\
&= F_{\alpha_m, \mu_m} (\phi).
\end{align*}
In short, we inherit the surjectivity from Lagrangian densities themselves.

This argument can be described efficiently in terms of the diagram.
The map $\int$ is surjective, so any element $F$ in $\loc_c(U)$ can be written as $\int \Lcal$ for some $\Lcal$ in $\llag^n_c(U)$.
But since $\llag^n_c$ is a cosheaf, 
$\Lcal$ is in the image of the horizontal map out of $\bigoplus_i \llag^n_c(U_i)$,
so pick $(\Lcal_i)_{i \in I}$ that maps to $\Lcal$ by the sum of extensions.
Then commutativity ensures that $c$ maps $(\int \Lcal_i)_{i \in I}$ to~$F$.

We now need to show that $\bar{c}$ is injective.
Let $(F_i)_{i \in I}$ be an element of $\bigoplus_i \loc_c(U_i)$ that $c$ sends to zero.
We need to show that this element is in the image of $\epsilon$. 
As $\int$ is surjective, pick an element $(\Lcal_i)_{i \in I}$ such that $\int \Lcal_i = F_i$ for every $i$.
We see that 
\[
\int \left( \sum_i \ext_i(\Lcal_i) \right) = 0
\]
by commutativity,
so there is some $\Fcal$ in $\llag^{n-1}_c(U)$ such that 
\[
\d_{\llag} \Fcal = \sum_i \ext_i(\Lcal_i)
\]
as the rightmost column is exact.
But this means we can replace the element $(\Lcal_i)_{i \in I}$ by an element $(\Lcal'_i)_{i \in I}$ such that $\sum_i \ext_i(\Lcal'_i) = 0$.
As the middle row is exact, we see that there is $(\Lcal_{jk})_{j,k}$ in $\bigoplus_{j,k} \loc_c(U_j \cap U_k)$ whose image under $\epsilon^n$ is $(\Lcal'_i)_{i \in I}$.
But commutativity of the diagram ensures that $\epsilon$ sends $(\int \Lcal_{jk})_{j,k}$ to~$(F_i)_{i}$.
\end{proof}

\subsubsection{A distributional variant}

It can be convenient to allow Lagrangians with values in distributional de Rham differential forms, rather than just smooth de Rham forms.
These are sometimes called de Rham~{\em currents}.

Let $\overline{\Omega}^*$ denote the distributional completion of the differential forms.
Concretely, it can be understood as
\[
\overline{\Omega}^* = \Dcal' \otimes_{C^\infty} \Omega^*,
\]
which is a sheaf and where here (for just this moment!) $\Dcal'$ denotes the continuous linear dual to $\Dcal = C^\infty_c$,
which is the classic notation from distribution theory.
Similarly, we can consider the compactly supported variant $\overline{\Omega}^*_c$, given by
\[
\overline{\Omega}^*_c = \Ecal' \otimes_{C^\infty_c} \Omega^*_c,
\]
which is a cosheaf and where here (for just this moment!) $\Ecal'$ denotes the continuous linear dual to~$\Ecal = C^\infty$.

Given a codimension $p$-current $\mu$ (i.e., a continuous linear functional on $n-p$-forms) and a Lagrangian $\alpha$,
there is a {\em Lagrangian $p$-current} that assigns to a field $\phi$, the codimension $p$-current $\alpha(\jet_k(\phi)) \mu$.
The space of Lagrangian $p$-currents $\overline{\llag}^p$ is the vector space spanned by such.
Our work above then immediately suggests the following definition.

\begin{df}
\label{df: dg local curr}
Let $\overline{\lloc}^\bullet$ denote the distributional completion of $\lloc^\bullet$. 
The compactly-supported sections $\overline{\lloc}^\bullet_c$ of this complex determine the {\em local differential currents} on the fields~$\Ecal$.
\end{df} 

Note that the zeroth cohomology encodes the distributional analogue of local functionals;
it is given by the distributional completion of Lagrangian densities up to total derivative.

\subsection{Fiberwise analytic local functionals}
\label{analytic variant}

Recall Remark~\ref{analyticity of loc} about working with fiberwise analytic Lagrangians. 
We now carefully explain what we mean and provide notations.

Given a vector bundle $\pi:V \to X$ on a smooth manifold, 
let ${\rm V}$ denote the total space of the bundle.
(Note the change from $V$ to ${\rm V}$.)
For any smooth function $f$ on ${\rm V}$,
its restriction $f|_{V_x}$ to a fiber $V_x = \pi^{-1}(x)$ is a function on a vector space,
and hence one can ask if it is a polynomial function or a real-analytic function.
We will say $f$ is {\em analytic along the fibers} (or fiberwise analytic) if the restriction to every fiber is real-analytic.
(Note that in a neighborhood $U$ of $x$ inside $X$ on which $V$ is trivialized, we can describe $f$ as a power series on the fiber $V_x$ whose coefficients vary smoothly over $U$.)
Let $C^{f\omega}({\rm V})$ denote the space of smooth functions on ${\rm V}$ that are fiberwise analytic.

In consequence an order $k$ Lagrangian $\alpha$ is {\em analytic along the fibers} if it lives in $C^{f\omega}({\rm J}^k E)$, where ${\rm J}^k E$ is the total space of the jet bundle $J^k E \to X$ ({\it cf.} definition~\ref{order k lag}).
Taking the colimit as $k$ goes to infinity, we obtain the fiberwise analytic Lagrangians.
Similarly, we obtain a special class of local functionals

\begin{df}
\label{def of fib analytic loc}
Let $\loc^{f\omega}_{(k)}$ denote the subsheaf of $\loc_{(k)}$ spanned by functionals $F_{\alpha,\mu}$ where
$\alpha$ is an order $k$, fiberwise analytic Lagrangian and $\mu$ is a smooth density.
Let $\loc^{f\omega}$ denote the colimit over $k$ of the~$\loc_{(k)}(X)$. 

Similarly, let $\loc^{f\omega}_c$ denote the subcosheaf of $\loc_c$ spanned by functionals $F_{\alpha,\mu}$ where $\alpha$ is a fiberwise analytic Lagrangian (of some finite but arbitrary order) and $\mu$ is a compactly-supported smooth density.
We call $\loc^{f\omega}_c(X)$ the {\em fiberwise analytic local functionals} on~$\Ecal(X)$.
\end{df}

We describe the dg version of fiberwise analytic local functionals in Section~\ref{fbrwsanalytic dg v} below.

An appealing feature of Definition~\ref{def of fib analytic loc} is that such functionals are analytic as functions on~$\Ecal(X)$.

\begin{prop}
Every fiberwise analytic local functional $F \in \loc^{f\omega}_c(X)$ is an analytic function on $\Ecal(X)$.
That is, $\loc^{f\omega}_c(X) \subset C^\omega(\Ecal(X))$.
\end{prop}

\begin{rem}
Later we will use this result to show that the S-matrix (or time-ordered exponential) is injective.
\end{rem}

\begin{proof}
If one views the functional
\[
F_{\alpha, \mu} (\phi) = \int_X \alpha(\jet_k \phi) \mu
\]
as a composition of maps
\[
\Ecal \xto{\jet_k} \Gamma(X, J^k E) \xto{\alpha} C^\infty(X) \xto{\mu \cdot -} \cDens_c(X) \xto{\int_X} \CC,
\]
then one verifies that each map is real-analytic (in fact, most are linear) and so the composition is real-analytic.
(See Chapter 11 of~\cite{Michor} for detailed discussion.)
\end{proof}

\subsubsection{The coproduct}
\label{fbrws cprdct}

Another appealing feature of fiberwise analytic functionals is a coproduct ${\vr \Delta}$ that we will use later.

First, note that for an analytic function $f$ on a vector space $V$,
we have 
\[
f(v + w) = \sum_{(f)} f_{(1)}(v) f_{(2)}(w) 
\]
where we use Sweedler notation in the sum.
This sum is typically infinite but admits a natural ordering by polynomial degree.
It may be helpful to unpack this formula in coordinates, so suppose $V \cong \RR^n$.
If we write the Taylor expansion of $f$ around $v$ as
\[
f(v+w) = \sum_{\mu \in \NN^n} a_\mu(v) w^\mu,
\]
where $\mu$ denotes a multinomial, 
we see that we can take $f_{(1)}(v) = a_\mu(v)$ and $f_{(2)}(w) = w^\mu$.
In this way there is a clear relationship between the Sweedler sum and the Taylor expansion.

Second, note that since a jet bundle $J^k E \to X$ is a vector bundle, 
we have fiberwise addition 
$${\rm add}: J^k E \oplus J^k E \to J^k E,$$ 
where the sum $\oplus$ here means the Whitney sum.
Hence any function $\alpha \in C^\infty({\rm J^k}E)$ pulls back to a function ${\rm add}^* \alpha \in C^\infty(J^k E \oplus J^k E)$.
The pullback of a fiberwise analytic function is again fiberwise analytic.  Note also that we can view 
$$C^\infty(J^k E \oplus J^k E)\cong C^\infty(J^k E)\,\widehat{\otimes}\, C^\infty(J^k E),$$ 
where the completion of the tensor product is done in the topology induced by the space on the left-hand side.

Define an operator ${\vr \Delta}_k:C^\infty( J^k E)\rightarrow C^\infty(J^k E \oplus J^k E)$ on an order $k$ fiberwise analytic Lagrangian $\alpha$ by
\[
{\vr \Delta}_k \alpha = {\rm add}^* \alpha
\]
so that when evaluated on prolongations of two fields, we have
\be\label{eq:coprod}
{\vr \Delta}_k \alpha\,(\jet_k\psi,\jet_k\phi)\doteq \alpha(\jet_k\psi+\jet_k \phi)=\sum_{(\alpha)} \alpha_{(1)}(\jet_k \psi)\alpha_{(2)}(\jet_k \phi)\,,
\ee
where we use Sweedler notation in the sum and which is explicitly given by the Taylor expansion.
These equations hold as functions on the manifold~$X$.
Note that if we view an order $k$ Lagrangian as an order $k+1$ Lagrangian, the behavior of this operator is unchanged:
\[
{\vr \Delta}_k\alpha = {\vr \Delta}_{k+1}\alpha.
\]
Hence we define an operator ${\vr \Delta}$ on all Lagrangians by setting ${\vr \Delta} \alpha\doteq {\vr \Delta}_k\alpha$ when $\alpha$ has order~$k$. 
(Note that we can use partitions of unity on $X$ to write $\alpha$ as a sum of Lagrangians, each of which has some bounded order.)
We call this operator the {\em coproduct} on Lagrangians.

Now consider a fiberwise analytic local functional of the form
\[
F(\phi) =\int_X \alpha(\jet_k \phi) \mu,
\]
where $\mu$ is a smooth density.
We see
\begin{align*}
F(\phi+\psi)
&= \int_X \alpha(\jet_k \phi+\jet_k \psi) \mu\\
&=\sum_{(\alpha)}\int_X \alpha_{(1)}(\jet_k\psi(x))\alpha_{(2)}(\jet_k \phi(x))\mu(x),
\end{align*}
thanks to the coproduct.

\subsubsection{The dg version}
\label{fbrwsanalytic dg v}

In all the definitions of Section~\ref{sec:local}, 
one can simply replace Lagrangians by fiberwise analytic Lagrangians,
and we offer the following modification of Definition~\ref{lag p-forms} as a representative example.

\begin{df}
The sheaf of order $k$ {\em fiberwise analytic \em Lagrangian $p$-forms} denotes 
\[
\llag^{f\omega ,\, p}_{(k)} = \pi^k_* C^{f\omega}_{{\rm J}^k E} \otimes_{C^\infty_X} \Omega^p_X,
\]
which $\Upsilon^{p}$ (recall ~\eqref{eqn Fp}) embeds into~${\rm Maps}(\Ecal, \Omega^p_X)$.

Let the {\em fiberwise analytic Lagrangian $p$-forms} $\llag^{f\omega,\,p}(X)$ denote the colimit over $k$ of the~$\llag_{(k)}^{f\omega, \, p}(X)$.
\end{df}

Note that this condition of fiberwise analyticity is local along the spacetime manifold,
and hence it does not affect local-to-global arguments (e.g., the fiberwise analytic Lagrangian $p$-forms form a sheaf on~$X$).

\section{Multilocal functionals}
\label{sec:mlocal}

We now define multilocal functionals, recapitulating the development in the preceding section.

\subsection{Multilocal functionals, the strict version}

In \cite{FR,FR3} a key class of observables is the {\em multilocal functionals},
which arise in a simple way.
Consider a scalar field theory on $\RR^n$, 
as in our example~\eqref{eq: rep of local} of a local functional.
A product of two such local functionals yields a concrete example of a quadratic multilocal functional:
\begin{align}
\label{eq: rep of multilocal}
(F \cdot G)(\phi) &= \int_{x \in \RR^n} \phi(x)^3  f(x)\, \d^n x \cdot \int_{y \in \RR^n}(\partial_1 \phi(y))^7 (\partial_1 \partial_n^2 \phi(y)) g(y)\, \d^n y \\
&= \int_{\RR^n \times \RR^n} \phi(x)^3 (\partial_1 \phi(y))^7 (\partial_1 \partial_n^2 \phi(y)) f(x) g(y)\, \d^n x \, \d^n y,
\end{align}
where $f$ and $g$ are compactly supported smooth functions.
Note that we rewrote the product of two integrals as an integral over the product space.

One way to encompass such examples is to take the symmetric algebra generated by the compactly supported local functionals,
where we mean the symmetric algebra in a purely algebraic sense:
one uses the algebraic tensor product, without any completions.
It is not obvious (to us) that this construction forms a factorization algebra.
Instead it seems more convenient to use a slight thickening that yields a factorization algebra but also works with Epstein-Glaser renormalization.

Our thickening of quadratic functionals amounts to replacing a product function $f(x) g(y)$ by an arbitrary smooth function $h(x,y)$ with compact support in $\RR^n \times \RR^n$.
Hence, an example of a quadratic multilocal functional is 
\[
H(\phi) = \int_{\RR^n \times \RR^n} \phi(x)^3 (\partial_1 \phi(y))^7 (\partial_1 \partial_n^2 \phi(y)) h(x,y)\, \d^n x \, \d^n y
\]
where $h(x,y)$ is smooth and has compact support in $\RR^n \times \RR^n$.
More generally, we complete $k$-fold products of Lagrangian densities by allowing multiplication by an element~$C^\infty_c((\RR^n)^k)$.

We now formalize this idea.

\begin{df}
\label{def:multilag}
Let $\lag_c^{[m]}$ denote the functor from $\Opens(X)$ to $\Vec$ where
\[
\lag_c^{[m]}(U) = \cinfty_c(U^m) \otimes_{\cinfty_c(U)^{\otimes m}} \lag_c(U)^{\otimes m}.
\]
(This tensor product is purely algebraic.) 
It encodes compactly supported {\em degree $m$ Lagrangian densities} on~$\Ecal(U)$.
\end{df}

This notion captures the integrands with which we work.
Given $\alpha \in \lag^{[m]}_c(U)$, we obtain a {\em degree $m$ multilocal functional} by
\begin{equation}
\label{eq: def of multiloc}
F_\alpha(\phi) = \int_{U^m} \alpha(\jet_\infty\phi(x_1), \ldots, \jet_\infty\phi(x_m)),
\end{equation}
that is, we evaluate the Lagrangian density on $m$ copies of the jet expansion of the field~$\phi$ and integrate over the $m$-fold product space.
In other words, there is a natural map
\begin{equation}
\label{eq: F multiloc}
q^{\langle m \rangle}: \lag_c^{[m]}(U) \to C^\infty(\Ecal(U))
\end{equation}
by~\eqref{eq: def of multiloc}, just as in Definition~\ref{def of loc}.
This notion is just the $m$-fold version of local functionals arising by integrating Lagrangian densities.

Note that total derivatives vanish along this map.
Note as well that this map factors through the quotient $\lag_c^{[m]}(U)/S_m$ by the symmetric group action.

\begin{df}
Let $\loc_c^{\langle m\rangle}$ denote the functor from $\Opens(X)$ to $\Vec$ that assigns {\em degree $m$ (or $m$-fold) multilocal functionals} on~$\Ecal(U)$ to each open set $U$.
It is the image of the map~\eqref{eq: F multiloc}.

Let $\mloc_c$ denote the functor from $\Opens(X)$ to $\Vec$ where
\[
\mloc_c(U) = \bigoplus_{m = 0}^\infty \loc_c^{{\langle m\rangle}}(U),
\]
which encodes {\em multilocal functionals} on~$\Ecal(U)$.
\end{df}

\begin{rem}
Given a functional $F \in \cinfty(\Ecal(X))$, 
it may not be easy to recognize if it is local or multilocal.
Developing recognition criteria is the central subject of \cite{BDGR};
see Theorem~I.2.
\end{rem}

It will be convenient to have an intermediary version of this construction.
Given $\alpha \in \lag^{[m]}_c(U)$, we obtain a functional on $m$-tuples of fields by
\begin{equation}
\label{eq: alt def of multiloc}
F_\alpha(\phi_1, \ldots, \phi_m) = \int_{U^m} \alpha(\jet_\infty\phi_1(x_1), \ldots, \jet_\infty\phi_m(x_m)).
\end{equation}
In other words, there is a natural map
\[
q^{[m]}: \lag_c^{[m]}(U) \to C^\infty(\Ecal(U)^m)
\]
by~\eqref{eq: alt def of multiloc}. 
We use $\loc^{[n]}_c(U)$ to denote the image of this map~$q^{[m]}$,
and we call such the {\em polarized} multilocal functionals of degree~$m$.

Observe that
\[
q^{\langle m \rangle}: \lag_c^{[m]}(U) \to \loc^{\langle m \rangle}_c(U) \subset C^\infty(\Ecal(U))
\]
factors as the composite
\[
\lag_c^{[m]}(U) \xto{q^{[m]}} \loc^{[m]}_c(U) \subset C^\infty(\Ecal(U)^m) \xto{\delta_m^*} C^\infty(\Ecal(U))
\]
where 
\[
\delta_m: \Ecal(U) \hookrightarrow \Ecal(U)^m
\]
sends a field $\phi$ to $m$ self-copies~$(\phi, \ldots, \phi)$.

\begin{lemma}\label{lem:multilocal}
The map $\loc_c^{[m]}/S_m \to \loc_c^{\langle m\rangle}$ is an isomorphism of cosheaves.
\end{lemma}

\begin{proof}
In theorem 3.1 of \cite{FR3}, it is shown that the map
\[
\Sym^m(\loc_c(U)) = \loc_c(U)^{\otimes m}/S_m \to \cinfty(\Ecal(U))
\]
is injective for any open set~$U$.
The argument carries over verbatim to the thickening $\loc_c^{[m]}/S_m$, 
as it amounts to studying $m$-fold functional derivatives and using a nifty trick with partitions of unity.
\end{proof}

With our basic ingredients in place, 
let us check that $\mloc_c$ forms a prefactorization algebra. 
As a first step, we need to describe the structure maps of higher arity.
First, consider the simplest situation.

\begin{lemma}
\label{lem: mloc has prefact maps}
For any pair of disjoint opens $U_0$ and $U_1$, there is a canonical map
\[
\iota_{U_0,U_1}: \mloc_c(U_0) \otimes \mloc_c(U_1) \to \mloc_c(U_0 \sqcup U_1).
\]
\end{lemma}

\begin{proof}
For any pair of disjoint open sets $U_0$ and $U_1$, there is a canonical isomorphism
\[
\loc_c(U_0) \oplus \loc_c(U_1) \to \loc_c(U_0 \sqcup U_1)
\]
because $\loc_c$ is a precosheaf and because of the universal property of the direct sum.
Taking the symmetric algebra sends direct sums to tensor products so
we obtain a canonical map
\begin{equation}
\label{sym map}
\Sym(\loc_c(U_0)) \otimes \Sym(\loc_c(U_1)) \cong \Sym(\loc_c(U_0) \oplus \loc_c(U_1)) \to \Sym(\loc_c(U_0 \sqcup U_1)).
\end{equation}
(Here $\Sym$ means the symmetric algebra using the algebraic tensor product.)
This construction offers the essential algebraic reason that desired map exists, because $\mloc_c$ can be seen as a thickening of $\Sym(\loc_c)$.
(That symmetric algebra is, in fact, dense inside~$\mloc_c$.)

Let's start by unwinding what each summand $\loc_c^{\langle n \rangle}(U_0 \sqcup U_1)$ means.
By definition it is a quotient of $\lag_c^{[n]}(U_0 \sqcup U_1)$,
and
\[
\lag_c^{[n]}(U_0 \sqcup U_1) = \cinfty_c((U_0 \sqcup U_1)^n) \otimes_{\cinfty_c(U_0 \sqcup U_1)^{\otimes n}} \lag_c(U_0 \sqcup U_1)^{\otimes n}.
\]
Each factor can be further expanded.
For example,
\[
(U_0 \sqcup U_1)^n = U_0^n \sqcup (U_0^{n-1} \times U_1) \sqcup \cdots \sqcup U_1^n = \bigsqcup_{ (v_1,\ldots,v_n) \in \{0,1\}^n} \prod_{j=1}^n U_{v_j},
\]
so we find
\[
\cinfty_c((U_0 \sqcup U_1)^n) \cong \prod_{ (v_1,\ldots,v_n) \in \{0,1\}^n} C^\infty_c \left( \prod_{j=1}^n U_{v_j} \right)
\]
as algebras. Similarly, because 
\[
\cinfty_c(U_0 \sqcup U_1) \cong \cinfty_c(U_0) \times \cinfty_c( U_1),
\]
we find
\[
\cinfty_c(U_0 \sqcup U_1)^{\otimes n} \cong \prod_{ (v_1,\ldots,v_n) \in \{0,1\}^n}  \bigotimes_{j=1}^n \cinfty_c( U_{v_j})
\]
We mean here the algebraic tensor product.
Now recall that there is a natural inclusion
\[
\cinfty_c(M) \otimes \cinfty_c(N) \hookrightarrow \cinfty_c(M \times N)
\]
for any manifolds $M$ and $N$, so that we have
\[
\prod_{v = (v_1,\ldots,v_n) \in \{0,1\}^n}  \bigotimes_{j=1}^n \cinfty_c( U_{v_j}) \hookrightarrow 
\prod_{(v_1,\ldots,v_n) \in \{0,1\}^n} C^\infty_c \left( \prod_{j=1}^n U_{v_j} \right).
\]
Note that the inclusion is on each factor $v = (v_1,\ldots,v_n)$ separately.

The same argument shows
\[
\lag_c(U_0 \sqcup U_1)^{\otimes n} \cong \prod_{ (v_1,\ldots,v_n) \in \{0,1\}^n} \lag_c \left( \prod_{j=1}^n U_{v_j} \right).
\]
It is a module over $\prod_{v = (v_1,\ldots,v_n) \in \{0,1\}^n}  \bigotimes_{j=1}^n \cinfty_c( U_{v_j})$ where the $v$th factor of the algebra acts on the $v$th component of the module built from~$\lag_c$.
Thus, applying these observations to the tensor product, we find
\begin{equation}
\label{lagcunwound}
\lag_c^{[n]}(U_0 \sqcup U_1) = \prod_{ (v_1,\ldots,v_n) \in \{0,1\}^n} C^\infty_c \left( \prod_{j=1}^n U_{v_j} \right) \otimes_{\bigotimes_{j=1}^n \cinfty_c( U_{v_j})}\lag_c \left( \prod_{j=1}^n U_{v_j} \right).
\end{equation}
In other words, we can decompose $\lag_c^{[n]}(U_0 \sqcup U_1)$ into components associated to the disjoint opens of the form~$\prod_{j=1}^n U_{v_j}$.

There is then a natural map 
\begin{equation}
\label{lag binary map}
\lag_c^{[k]}(U_0) \otimes \lag_c^{[n-k]}(U_0) \to \lag_c^{[n]}(U_0 \sqcup U_1) 
\end{equation}
for every natural number $0 \leq k \leq n$, as follows.
On the left hand side we have the opens $(U_0)^k$ and $(U_1)^{n-k}$ while on the right hand side we have $(U_0)^k \times (U_1)^{n-k}$ as a component of $(U_0 \sqcup U_1)^n$.
Using the expression~\eqref{lagcunwound} for each appearance of $\lag_c^{[*]}$,
one obtains a map
\begin{align*}
( \cinfty_c(U_0^k) &\otimes_{\cinfty_c(U_0)^{\otimes k}} \lag_c(U_0)^{\otimes k} )
\otimes 
\left(\cinfty_c(U_1^{n-k}) \otimes_{\cinfty_c(U_1)^{\otimes n-k}} \lag_c(U_1)^{\otimes n-k}\right)\\
&\xto{\cong}
(\cinfty_c(U_0^k) \otimes \cinfty_c(U_1^{n-k})) 
\otimes_{\cinfty_c(U_0)^{\otimes k} \otimes \cinfty_c(U_1)^{\otimes n-k}} 
\left(\lag_c(U_0)^{\otimes k} \otimes \lag_c(U_1)^{\otimes n-k} \right)
\end{align*}
and the codomain is a component of $\lag_c^{[n]}(U_0 \sqcup U_1)$, as desired.

The map $\lag_c^{[m]} \to \loc_c^{[m]}$ is given by evaluating a Lagrangian density on fields
(i.e., looking at the associated function on fields).
This map~\eqref{lag binary map} induces the needed map 
\[
\loc_c^{[k]}(U_0) \otimes \loc_c^{[n-k]}(U_0) \to \loc_c^{[n]}(U_0 \sqcup U_1) 
\]
because evaluating the ``product of observables'' (the output of~\eqref{lag binary map}) on a field agrees with the product of each observable evaluated on the field.

Running over all $k$ and $n$, we obtain the desired map for~$\mloc_c$.
\end{proof}

We use this construction to define the structure map
\[
\iota_{\{U_i\};V}: \bigotimes_{i} \mloc_c(U_i) \to \mloc_c(V)
\]
for any finite tuple $U_1$, \dots, $U_k$ of pairwise disjoint opens contained in a common open $V$ as the composite
\[
\bigotimes_{i} \mloc_c(U_i) \to \mloc_c(\sqcup_i U_i) \to \mloc_c(V)
\]
where the first map arises from the lemma (just induct on the number of opens) and the second map is the map given by viewing $\mloc_c$ as a precosheaf.
This construction is manifestly equivariant in relabelings of the tuple, 
and compositions of the structure maps are associative on the nose.
(See Chapter 3 of \cite{CoGw} for a full treatment of this kind of construction.)
In consequence we have shown the following.

\begin{lemma}
The construction $\mloc_c$ determines a prefactorization algebra on~$X$.
\end{lemma}

We now turn to verifying $\mloc_c$ is a cosheaf with respect to the Weiss topology, 
which is the remaining condition to be a factorization algebra.

The crucial reduction step is to observe that the structure maps of the precosheaf $\mloc_c$ preserve the direct sum decomposition into ``symmetric degrees'' (i.e., into the components $\loc_c^{\langle m \rangle}$).
For $U \subset U'$, the extension sends an element in $\loc_c^{\langle m \rangle}(U)$ to an element in $\loc_c^{\langle m \rangle}(U')$.
Hence it suffices to check the gluing axiom for each summand $\loc_c^{\langle m \rangle}$ separately.
We will do this in stages.

Consider first the polarized multilocal functionals $\loc_c^{[m]}$,
because we can then make a computation on the product space~$X^m$.

\begin{lemma}
The functor $\loc_c^{[m]}$ from $\Opens(X^m)$ to graded vector spaces is a cosheaf.
\end{lemma}

The proof of lemma~\ref{lem: loc is cosheaf} carries over to $\loc_c^{[m]}$ with very minor modifications.
One first verifies that $\lag_c^{[m]}$ is a cosheaf on $X^m$ (by a partition of unity argument) 
and then verify it induces the desired behavior of~$\loc_c^{[m]}$ 
(because colimits commute, a colimit of cosheaves is the underlying colimit of precosheaves).

Note that as $\loc^{\langle m \rangle} = \loc_c^{[m]}/S_m$, and colimits of cosheaves are computed as precosheaves, we immediately obtain the following.

\begin{cor}
Because $\loc_c^{[m]}$ is a cosheaf on $X^m$, so is $\loc^{\langle m \rangle}$ is too.
\end{cor}

We now have all the ingredients to swiftly prove our main goal for this section.

\begin{prop}
\label{prop: mloc is fact alg}
The construction $\mloc_c$ determines a factorization algebra on~$X$.
\end{prop}

We mean here a {\em strict}, not a homotopy, factorization algebra.

\begin{proof}
Let $\{U_i\}_i$ be a Weiss cover of the open subset $U \subset X$.
For any natural number~$m$,
observe that $\{U_i^m\}_i$ forms an ordinary cover for $U^m$:
for any point $(x_1,\ldots,x_m) \in U^m$, 
the finite subset $\{x_1,\ldots,x_m\} \subset U$ is contained in some $U_j$ by the Weiss condition,
so that we know $(x_1,\ldots,x_m) \in U^m_j$.
Thus, for each $m$, we know that $\loc_c^{\langle m \rangle}(U^m)$ can be recovered as the cokernel of the structure maps arising from the cover $\{U_i^m\}_i$ of $U^m$.
As $\mloc_c(U)$ is the direct sum (i.e., a colimit) of the $\loc_c^{\langle m \rangle}(U^m)$,
we see that it is cokernel of the structure maps arising from the Weiss cover~$\{U_i\}_i$.
\end{proof}

This result, however, is not ideal for our purposes, 
because it is not homotopical and hence does not play nicely with the BV formalism.

\subsection{Multilocal functionals, the dg version}

We swiftly combine the maneuvers of the preceding two subsections.

A degree $m$ Lagrangian $p$-form consists of a function on the jets of a field with values in $p$-forms on the product manifold~$X^m$ .
To formulate this notation precisely,
we note that there is a natural vector bundle
\[
{\rm J^k} E^{[m]} := \left( {\rm J}^k E \right)^{\boxtimes m} \to X^m
\]
given by the outer tensor product,
and that every $m$-tuple of fields $\phi_1, \ldots, \phi_m$ in $\Ecal(X)$ determines a section of this bundle by taking the tensor product of the prolongations
\[
\jet_k(\phi_1)(x_1) \boxtimes \cdots \boxtimes \jet_k(\phi_m)(x_m)),
\]
where $(x_1,\ldots,x_m) \in X^m$.

\begin{df}
The presheaf on $X$ of order $k$ {\em degree $m$ Lagrangian $p$-forms} assigns to an open set $U$,
the vector space
\[
\llag_{(k)}^{[m]p}(U) = \pi_* C^\infty_{{\rm J^k} E^{[m]}}(U^m) \otimes_{C^\infty(U^m)} \Omega^p(U^m).
\]
Let the {\em degree $m$ Lagrangian $p$-forms} $\llag^{[m]p}(U)$ denote the colimit over $k$ of the~$\llag_{(k)}^{[m]p}(U)$.
Let $\llag_c^{[m]p}(U)$ denote the degree $m$ Lagrangian $p$-forms with compact support.
\end{df}

The de Rham differential sends degree $m$ Lagrangian $p$-forms to degree $m$ Lagrangian $p+1$-forms
by the same argument as Lemma~\ref{lem: d makes sense}.
(It can increase the order from $k$ to~$k+1$.)
Let $\d_\llag$ denote this differential, with the degree $m$ not displayed in the notation.
Hence, just as in the local case, we obtain a cochain complex $\llag^{[m]\bullet}$ given by
\[
\llag_c^{[m]0} \xto{\d_\llag} \llag_c^{[m]1} \xto{\d_\llag} \cdots \xto{\d_\llag} \llag_c^{[m](mn)}
\]
and concentrated in degrees $-mn$ to 0,
where $\d_\llag$ denotes the differential induced by the exterior derivative as just explained.
As in the local case, we take a mapping cone to remove the ``constant'' terms.

\begin{df}
Let $\lloc_c^{[m]\bullet}$ denote the precosheaf of cochain complexes of {\em degree $m$ compactly supported local differential forms}
\[
\operatorname{Cone}\left(\Omega^\bullet_{X^m, c}[mn] \xto{\iota} \llag_c^{[m]} \right)
\]
where $\iota$ includes a degree $m$ Lagrangian $p$-form as itself.

We call $\lloc_c^{[m]\bullet}$ the {\em degree $m$ local differential forms} with compact support.
\end{df}

Note that $\llag^{[m](mn)}$ is precisely $\lag^{[m]}$.
Hence the composite map
\[
\lloc_c^{[m]\bullet}(U) \xto{\tau_{\geq 0}} \llag_c^{[m](mn)}(U) = \lag_c^{[m]}(U) \xto{q^{\langle m \rangle}} C^\infty(\Ecal(U))
\]
realizes a degree $m$ Lagrangian top form as a functional on fields.

We now assemble these to define a cochain complex of local differential forms.

\begin{df}
\label{df: dg multiloc}
Let $\mlloc_c$ denote the functor from $\Opens(X)$ to $\Ch$ that assigns to an open set $U$, 
the cochain complex
\[
\mlloc_c(U) = \bigoplus_{m=0}^\infty \lloc_c^{[m]\bullet}(U)
\]
given by the direct sum of the degree $m$ compactly supported local differential forms.
We call $\mlloc_c$ the {\em multilocal differential forms}.
\end{df}

We note that one can replace de Rham forms everywhere by de Rham currents to produce $\overline{\mlloc}_c$,
the {\em multilocal differential currents},
which are the multilocal analogue of Definition~\ref{df: dg local curr}.

We now discuss how to equip $\mlloc_c$ with a prefactorization structure,
and the argument follows the pattern used for~$\mloc_c$.
Consider the simplest case of two disjoint opens.

\begin{lemma}
For any pair of disjoint opens $U$ and $U'$, there is a canonical morphism
\[
\iota_{U,U'}: \mlloc_c(U) \otimes \mlloc_c(U') \to \mlloc_c(U \sqcup U').
\]
\end{lemma}

Note that this map is a quasi-isomorphism when $U$ and $U'$ are (diffeomorphic to) balls,
by combining Theorem~\ref{olver thm} and Lemma~\ref{lem: mloc has prefact maps}.

\begin{proof}
It suffices to define the map on an arbitrary element of the form $\alpha \otimes \alpha'$, 
where $\alpha$ is a degree $m$ Lagrangian $p$-form on $U$ and $\alpha'$ is a degree $m'$ Lagrangian $p'$-form on $U'$.
By definition, $\alpha$ lives on the manifold $U^m$ and $\alpha'$ lives on the manifold $(U')^{m'}$.
On the product space $U^m \times (U')^{m'}$, we can take the wedge product $\alpha \wedge \alpha'$,
where we mean here the pullback of $\alpha$ from $U^m$ to the product space (and likewise for $\alpha'$).
Note that $U^m \times (U')^{m'}$ is a component of $(U \sqcup U')^{m+m'}$,
and so we can pushforward this wedge product to obtain a $m+m'$ Lagrangian $p+p'$-form on~$U \sqcup U'$.

More concisely, the argument from Lemma~\ref{lem: mloc has prefact maps} can be mimicked to construct the desired map for~$\mlloc_c$.
\end{proof}

We use this construction to define the structure map
\[
\iota_{\{U_i\};V}: \bigotimes_{i} \mlloc_c(U_i) \to \mlloc_c(V)
\]
for any finite tuple $U_1$, \dots, $U_k$ of pairwise disjoint opens contained in a common open $V$ as the composite
\[
\bigotimes_{i} \mlloc_c(U_i) \to \mlloc_c(\sqcup_i U_i) \to \mlloc_c(V)
\]
where the first map arises from the lemma and the second map is the canonical map because $\mlloc_c$ is a functor.
This construction is manifestly equivariant in relabelings of the tuple, 
and compositions of the structure maps are associative on the nose.
(See Chapter 3 of \cite{CoGw} for a full treatment of this kind of construction.)
In consequence we have shown the following.

\begin{lemma}
The construction $\mlloc_c$ determines a prefactorization algebra on~$X$ with values in~$\Ch$.
\end{lemma}

It remains to show that $\mlloc_c$ is a homotopy cosheaf on $X$ with respect to the Weiss topology.
The first step is to observe that the degree $m$ local differential forms $\lloc^{[m]\bullet}$ are a homotopy cosheaf on~$X^m$,
by directly mimicking the proof of Lemma~\ref{lem: lloc is homotopy cosheaf}.
The second step is to notice that a Weiss cover $\{U_i\}_i$ of an open subset $U \subset X$ determines an open cover $\{U_i^m\}_i$ for each space~$X^m$.
As $\lloc_c^{[m]\bullet}$ satisfies homotopy codescent on $U^m$ with respect to this cover,
and as $\mlloc_c$ is the direct sum (i.e., colimit) of the $\lloc_c^{[m]\bullet}$,
we obtain the following.

\begin{prop}
The construction $\mlloc_c$ determines a homotopy factorization algebra on~$X$.
\end{prop}

\begin{rem}
This proposition is an analog of Theorem 5.2.1 of Chapter 6 of \cite{CoGw},
which uses a different model of observables but is similar in spirit.
\end{rem}

We posit that this dg version $\mlloc_c$ of multilocal functionals is typically better to use than $\mloc_c$,
as we have shown that $\mlloc_c$ forms a factorization algebra in the homotopically correct sense.
It is also convenient for many physical purposes because it naturally contains Lagrangian $p$-forms,
and hence can efficiently encode currents and their generalizations.

\subsection{Fiberwise analyticity}
\label{analyticity for mloc}

We now remark upon a variant of the multilocal functionals and multilocal different forms:
restrict to those that are analytic functions along the fibers of jet bundles.
In practice this restriction does not affect most physical questions, 
as most observables of interest to physicists satisfy this condition,
and it is essential for the renormalization procedure we use for quantization.

We have already defined carefully fiberwise analytic local functionals and local differential forms in sections~\ref{analyticity of loc} and~\ref{fbrwsanalytic dg v}.
The multilocal version admits an immediate analog.

\subsubsection{The coproduct}
\label{fbrws cprdct ml}

The coproduct introduced in section \ref{fbrws cprdct} extends to the multilocal case. 
Consider the vector bundle 
\[
J^{k_1} E \boxplus\dots\boxplus J^{k_m} E \to X^m.
\]
To make the notation more compact, let ${\vr k}\equiv (k_1,\dots,k_m)$ and let $J^{\vr k} E^{\boxplus m} \equiv J^{k_1} E \boxplus\dots\boxplus J^{k_m} E$.
Again we have the fiberwise addition
$${\rm add}: (J^{\vr k} E^{\boxplus m}\oplus J^{\vr k} E^{\boxplus m}) \to J^{\vr k} E^{\boxplus m}\,,$$
where the sum $\oplus$ here means the Whitney sum over~$X^m$.
Define an operator 
$${\vr \Delta}_{\vr k}:C^\infty(J^{\vr k} E^{\boxplus m}) \rightarrow C^\infty(J^{\vr k} E^{\boxplus m}\oplus J^{\vr k} E^{\boxplus m})$$ 
by ${\vr \Delta}_{\vr k} \alpha = {\rm add}^* \alpha.$
This map manifestly induces an operator ${\vr \Delta}_{\vr k}$ on $\lag_{(k_1)}\otimes \dots\otimes  \lag_{(k_m)}$,
and so we also obtain an operator $\vr \Delta$ on~$\lag^{[m]}_c$,
analogously to the local case.

For a fiberwise analytic functional on $m$-tuples of fields of the form
\[
F_\alpha(\phi_1, \ldots, \phi_m) = \int_{X^m} \alpha(\jet_{k_1}\phi_1(x_1), \ldots, \jet_{k_m}\phi_m(x_m))\,,
\]
where $\alpha\in \lag^{[m]}_c$ is fiberwise analytic, we can write
\begin{align*}
F_\alpha(\phi_1+\psi_1, \ldots, \phi_m+\psi_m)
&= \int_{X^m} \alpha(\jet_{k_1} \phi_1+\jet_{k_1} \psi_1,\dots, \jet_{k_m} \phi_m+\jet_{k_m} \psi_m) \\
&=\sum_{(\alpha)}\int_{X^m} \alpha_{(1)}(\jet_{k_1}\psi_1(x_1),\dots, \jet_{k_m}\psi_m(x_m))\alpha_{(2)}(\jet_{k_1} \phi_1(x_1),\dots, \jet_{k_m} \phi_m(x_m)))
\end{align*}
using the coproduct.

\section{Renormalization}\label{sec:ren}

In this section we construct the renormalized time-ordering map $\TT: \mloc_c \to \overline{\mloc}_c$.
This construction plays a crucial role in constructing the quantum observables in Section~\ref{sec obsq}.

\begin{NB}
{\it In this section we work with fiberwise analytic observables}. 
We drop the superscript $f \omega$ in this section; it is left implicit.
\end{NB}

Our approach to $\TT$ follows the standard pattern in the recent literature on perturbative algebraic quantum field theory (with~\cite{FR3} as an explicit representative).
As a first step, observe that any multilocal functional can be written as a sum of $n$-fold multilocal functionals, so it suffices to specify  \textit{$n$-fold time-ordered products} of multi-local functionals
\[
\TT_n: \loc_c^{[n]} \to \overline{\mloc}_c.
\]
The direct sum of these operations provides $\TT$.
(Typically $\TT_n$ is defined on the $n$-fold product of local functions, 
so in this paper we must explain how to extend to the thickened multilocal functionals we just introduced.)
When constructing these maps,
one constructs, in fact, a map 
\[
\Tb_n: \lloc_c^{[n]} \to \overline{\mlloc}_c
\]
at the level of multilocal differential forms.
This feature of the construction is manifest from a close reading of the literature (we mainly follow \cite{BDF} and \cite{H}),
but it is not examined so explicitly in it,
so we foreground that aspect here.

Before we construct the time-ordered product, 
we need to recall some  ingredients from the quantization of free theories.

\subsection{The pAQFT approach to free quantum theories}
\label{sec:freequantumtheory}

Following \cite{BDF, Book}, we assume that equations of motion are of the form
\[
P\ph=0
\] 
where $P$ is a Green hyperbolic operator~\cite{GreenBear}.
That means, that $P:\Ecal\rightarrow \Ecal^!$ 
is a hyperbolic differential operator from the fields $\Ecal$  to ``antifields''\footnote{More on $\Ecal^!$ and interpretation as antifields can be found in Section~\ref{sec: symplectic}}
$\Ecal^! = \Gamma(X, E^* \otimes \Dens)$ and it has unique advanced and retarded Green functions $\Delta^{\rm A/R}$. In section \ref{sec:classobs} we will explain how to obtain such and operator in gauge theories using the BV formalism.

We define the {\em Pauli-Jordan function} by
\[
\Delta=\Delta^{\rm R}-\Delta^{\rm A}\,.
\]
It determines the {\em Peierls bracket} on multilocal functions~by
\begin{equation}\label{eq:Peierls}
\Pei{F}{G} =\left<\delta F,\Delta \delta G\right>
\end{equation}
where $\delta F$ denotes the differential of $F$, i.e., a 1-form on $\Ecal$ (see Section~\ref{sec: dynamics} for more discussion).
This functional $\Pei{F}{G}$ may  not be multilocal, however,
so we will impose conditions to obtain a well-behaved Poisson algebra.
The key point is to control the singularity structure of functional derivatives,
so we impose conditions on the wavefront sets of the multilocal functionals.

Consider the following subset of~$T^* X^n$:
\[
\Xi_n^c \doteq \left\{ (x_1,\dots,x_n;k_1,\dots,k_n)\, \Big|\, (k_1,\dots,k_n)\in (\overline{V}_+^n \cup \overline{V}_-^n)_{(x_1,\dots,x_n)}
\right\}\,,
\]
where $(\overline{V}_{\pm})_x$ denotes the closed future (+) or past (-) lightcone, understood as a conic subset of each tangent space~$T^*_xX$.
This region is where singularities of, say, the various Green functions live
(and which thus are sources of divergences in naive computation of Feynman integrals).
Hence we pick out a subclass of functionals to avoid issues along this subset.

\begin{df}
\label{def: microcausal}
An $n$-fold distributional multilocal functional $$F\in \overline{\loc}_c^{[n]}(X) \subset \Ci(\Ecal(X)^n,\RR)$$ is {\em microcausal} if 
\be\label{mlsc}
\WF(F)\subset \Xi_n \doteq T^* X^n \setminus \Xi_n^c
\ee
where $\WF$ denotes the wavefront set.

Let $\overline{\loc}^{[n]}_{\mu c}(X)$ denote the subspace of microcausal $n$-fold multilocal functionals.
\end{df}

Note that $\loc^{[n]}_c(X) \subset \overline{\loc}^{[n]}_{\mu c}(X)$,
as multilocal functionals have smooth integrands.

Microcausality is a microlocal condition, and hence it plays nicely with support.
One sees immediately, for instance,
that there is a precosheaf  $\overline{\loc}^{[n]}_{\mu c}$ by taking this subspace on each open in~$X^n$.
With a little care, one can show it is a cosheaf.
Hence, by repeating the constructions of Section~\ref{sec:mlocal},
one can show the following.

\begin{lemma}
The microcausal functionals $\overline{\mloc}_{\mu c} \subset \overline{\mloc}_{c}$ are a factorization algebra on~$X$.
\end{lemma}

The microcausal functionals are closed under fieldwise multiplication (meaning $(FG)(\phi) = F(\phi)G(\phi)$) and the Peierls bracket \cite{BFR}.
It is natural to try to quantize by defining a $\star$-product as a deformation of the Poisson algebra of microcausal functionals.
Following the literature \cite{H,FR3}, we define
\begin{equation*}
F\star G\doteq m\circ \exp({i\hbar D_W})(F\otimes G),
\end{equation*}
where $m$ is the fieldswise multiplication operator and  $D_W$ is the differential operator (on functionals) defined by 
\begin{equation*}
D_W\doteq \frac{1}{2} \sum_{\al, \beta} \left<{W}^{\al\beta},\frac{\delta^l}{\delta\ph^\al} \otimes \frac{\delta^r}{\delta\ph^\beta}\right>\,.
\end{equation*}
Here $W$ denotes the \textit{2-point function of a Hadamard state}: 
we require that $W$ is positive definite, 
it satisfies \textit{$W=\frac{i}{2}\Delta+H$} with $H$ a symmetric bisolution for~$P$, and its wavefront set has the form\
\[
\WF(W)=\left\{(x,k;x',-k')\in\dot{T}^*X^2
\, \big| \,(x,k)\sim(x',k'), k\in (\overline{V}_+)_x\right\}\,,
\]
where $\dot{T}^*X^2$ means the complement of the zero section in $T^*X^2$ and  $(x,k)\sim(x',k')$ means that there exists a null geodesic connecting $x$ and $x'$ and vectors $k$ and $k'$ are cotangent to it. In the definition of $D_W$, we distinguish between the left and right derivative since $\Ecal$ could be graded. 



\subsection{Time-ordered products}
\label{sec:loc:ren}
From the early work of Dyson, the perturbative construction of interacting theories involves the use of a time-ordered product of operators,
which appears in constructing the S-matrix.
We now review how this idea is implemented rigorously in pAQFT.
In the literature one finds a set of axioms of the renormalized $n$-fold time-ordered products,
but the construction (guaranteeing the existence of such products) is at the level of distributions on~$X^n$.
We state and motivate the usual axioms,
explain parallel axioms at the level of multilocal differential forms,
and then explain how to work at the level of distributions. This review is based mainly on \cite{BDF,HW,H}.

\subsubsection{Time-ordered products of local functionals}

The existence of renormalized $n$-fold time-ordered products has an inductive proof:
given the existence of $k$-fold time-ordered products $\TT_k$ for $k<n$, 
one shows the existence of $\TT_n$.
At each stage we impose certain useful and natural requirements that we now describe and motivate.
These conditions guarantee that the $\TT_n$s assemble into the kind of time-ordered product $\T$ expected in QFT.

To get the induction started, we fix the initial two steps with the following axiom.

\begin{enumerate}[{\bf (T 1)}] 
\item{\bf Starting element}:  $\TT_0=1$ and $\TT_1=\id$.
\end{enumerate} 	
(The condition on $\TT_1$ really means that it is the obvious inclusion $\loc_c \hookrightarrow \overline{\loc}_c$.)

As the time-ordered product should be related to the $\star$-product,
we will restrict their image to land in microcausal functionals. 
(Here we use the formulation found in \cite{H}, but the reader is encouraged to look at older formulations, e.g.,~\cite{BF0,HW01}.)

\begin{enumerate}[{\bf (T 1)}] 
\addtocounter{enumi}{1}
\item{\bf Microlocal spectrum condition}: We have $\TT_n: \loc_c^{[n]} \to \overline{\mloc}_{\mu c}$ for all~$n$.
\end{enumerate}

We now formulate a causality relation,
which is an important requirement on how the time-ordered product~$\T$ should relate to the non-commutative star product~$\star$ of the free quantum theory.

\begin{df}
Given subsets $U,V\subset X$,
then $U$ is {\em not later than} $V$, denoted
$U \preceq V$,
if and only if $U$ does not intersect the future of $V$, i.e. $U\cap J^+(V)=\emptyset$. 
\end{df}

We define $\preceq$ on local functionals by setting $F\preceq G$ if and only if $\supp(F)\preceq \supp (G)$.

Then the desired {\em casuality relation} for time ordered products is
\be\label{eq:causality}
F\T G=\left\{ \begin{array}{ll}
F\star G&\textrm{ if }G\preceq F\\
G\star F&\textrm{ if }F\preceq G
\end{array} \right.\,.
\ee
For the goals of this paper, this relation is crucial for reconstructing the net of observables $\fA$ from the factorization algebra, which we will obtain using~$\T$. 

More generally, one can extend the relation $\preceq$ to~$\loc_c^{[n]}$. 

\begin{df}
For $F\in \loc_c^{[m]}$ and $G\in \loc_c^{[k]}$,
the functional $F$ is {\em not later than G}, written $F\preceq G$, if and only if for all $i=1,\dots,m$, $j=1,\dots, k$, 
we have  $$\pi_i(\supp(F))\preceq \pi_j(\supp(G)),$$ 
where $\pi_i$, $\pi_j$ denote the projections from $X^m$ and $X^k$, respectively, to the $i$th and $j$th factors of~$X$.
\end{df} 

We denote the space of such causally-ordered pairs of multilocal functionals by
\[
\loc_c^{[k]}\preceq \loc_c^{[n-k]}\equiv \left\{(F,G)\in \loc_c^{[k]}\times \loc_c^{[n-k]} |\, F\preceq G\right\}\,.
\]
Relation~\eqref{eq:causality} suggests we impose the following axiom.

\begin{enumerate}[{\bf (T 1)}]
\setcounter{enumi}{2}
\item{\bf Causal factorization property}: There is a commuting diagram
\[
\begin{tikzcd}[column sep=large] \loc_c^{[k]}\preceq \loc_c^{[n-k]}\arrow[r,"\Tcal_{n-k}\times\Tcal_{k}"]\arrow[d,"\cdot"] &\overline{\mloc}_{\mu c} \times \overline{\mloc}_{\mu c}\arrow[d,"\star"]\\ \loc_c^{[n]}\arrow[r,"\Tcal_n"] & \overline{\mloc}_{\mu c}
\end{tikzcd}
\]
for all $k < n$.
In formulas,
\[
\TT_n(GF)=\TT_{n-k}(G)\star \TT_{k}(F)\,.
\]
given $(F,G)\in \loc_c^{[k]}\preceq \loc_c^{[n-k]}$.
\end{enumerate}	

This causal factorization property fixes the time-ordered product $\TT_n$ away from the small diagonal, inductively.

The permutation action of the symmetric group $S_n$ on $X^n$ induces an $S_n$-action $\rho$ on $\loc_c^{[n]}$. 
To ensure the graded commutativity of~$\T$, we impose the following condition.

\begin{enumerate}[{\bf (T 1)}] 
\setcounter{enumi}{3}
\item {\bf Symmetry}: For all $n\in\NN$ and $\sigma\in S_n$,
we have a commuting diagram
\[
\begin{tikzcd}[column sep=large] \loc_c^{[n]}\arrow[rr,"\rho(\sigma)^*"]\arrow[dr,"\TT_n"] &&\loc_c^{[n]} \arrow[dl,"\TT_n"]\\ 
& \overline{\mloc}_{\mu c} &
\end{tikzcd}
\]
In formulas, we require  $\rho(\sigma)^*\Tcal_n=\Tcal_n$.
\end{enumerate}	

Note that this axiom explains how to define the time-ordered product $\TT_n$ on $\loc_c^{\langle n \rangle}$.
Recall from Lemma~\ref{lem:multilocal} that we have a canonical isomorphism
\[
\loc_c^{[n]}/S_n \xto{\cong} \loc_c^{\langle n \rangle},
\]
so if we apply the inverse and then use the symmetrization map
\[
\loc_c^{[n]}/S_n \to \loc_c^{[n]},
\]
we can apply $\TT_n$ to $\loc_c^{\langle n \rangle}$.
We will denote this map by $\TT_n$ as well,
as it is unambiguous.
Thus we can introduce the time-ordering map.

\begin{df}
\label{df: T map}
The \emph{time-ordering map}
\[
\TT: \mloc_c \to \overline{\mloc}_{\mu c}
\]
is the direct sum of the maps $\TT_n$ over all~$n$.
\end{df}

Next we discuss how the maps $\TT_n$ should depend on fields. The space of fields $\Ecal$ is a vector space, and hence it acts on itself by translation.
For any function $F$ on $\Ecal$ and a field $\psi \in \Ecal$, 
let $F^\psi$ denote the functional
\[
F^\psi(\phi) = F(\phi + \psi),
\]
i.e., the pullback of $F$ along translation by $\psi$. Similarly, for $F\in \loc^{[n]}_c$, we define 
\[
F^\psi (\ph,\dots,\ph)\doteq F(\ph+\psi,\dots, \ph+\psi).
\]
We then require the following.

\begin{enumerate}[{\bf (T 1)}] 
\setcounter{enumi}{4}
\item {\bf Field independence}: There is a commuting diagram
\[
\begin{tikzcd}[column sep=large] \loc_c^{[n]}\arrow[r,"(\bullet )^{\psi}"]\arrow[d,"\Tcal_n"] &\loc^{[n]}_c \arrow[d,"\Tcal_n"]\\ 
\overline{\mloc}_{\mu c}\arrow[r,"(\bullet )^{\psi}"] & \overline{\mloc}_{\mu c} 
\end{tikzcd}
\]
for every~$n$.
In formulas, 
\begin{equation}
\label{FiedIndep}
\TT_{n}(F)^\psi = \TT_{n}(F^\psi)\,
\end{equation}
for any $F\in \loc^{[n]}_c$ and any field $\psi\in \Ecal$.
\end{enumerate}	

This axiom has consequences for derivatives. 
Let $D_{\phi_i}$ denote the functional derivative with respect to the variable $\phi_i$, $i=1,\dots,n$ of $F\in \loc_c^{[n]}$. 
The field independence axiom implies that
\[
\delta_\ph \left(\TT_{n
}(F) \right)=\sum_{i=1}^n\TT_{n}(\delta_{\phi_i}F).
\]
One can iteratively apply this fact to describe higher-order derivatives and hence a Taylor expansion. 

The next axiom guarantees that at any given order in $\hbar$, every $\Tcal_n$ depends only on finitely many terms in the Taylor series expansion of its argument. 
Let $F^{[N]}$ denote the truncation of the Taylor series of $F$ at order $N$ for each of the variables.

\begin{enumerate}[{\bf (T 1)}] 
\setcounter{enumi}{5}
\item {\bf $\ph$-Locality}: For all natural numbers $n$ and $N$, we have $$\TT_{n}(F)=\TT_{n}(F^{[N]})+\Ocal(\hbar^{N+1})$$
for every $F \in \loc_c^{[n]}$.
\label{PhLoc}
\end{enumerate}

We also want $\TT$ to be natural (i.e., to intertwine with inclusion of open subsets) 
so we impose the following.

\begin{enumerate}[{\bf (T 1)}] 
\setcounter{enumi}{6}
\item {\bf Covariance}: $\TT_n$ is a natural transformation from the functor $\loc_c^{[n]}$ to $\overline{\mloc}_{\mu c}$ for all~$n$.
\end{enumerate}	

\begin{rem}
Note that in contrast to \cite{BFV}, here $\loc_c^{[n]}$ and $\overline{\mloc}_{\mu c}$ are functors from the category of \textit{all opens in X} with the usual morphisms, rather than the category of causally convex subsets with the causality-preserving embeddings. 
(This category is typically denoted by $\loc$, which we realize might lead to collision of notation.)
\end{rem}

A final axiom is typically included to guarantee that the S-matrix be unitary,
which is a natural physical requirement.
In the constructive existence proofs,
it is found that one can guarantee such unitarity, 
but it is a feature orthogonal to the main focus of this paper, 
so we will not include this axiom here.

These axioms capture what we want from the $\TT_n$, 
but we need to show such time-ordered products {\em exist}. 
To do this, we posit a special form for the $\TT_n$ in terms of maps on Lagrangian densities, 
as is usual in the literature \cite{EG,Due19,HW01,HW02}. 
Following Hollands \cite{H}, it is natural to extend these maps to local differential forms.

\subsubsection{Time-ordered products of local differential forms}

We aim to construct a degree-preserving map
\[
\mathbb{T}_n: \lloc_c^{[n]}\rightarrow  \overline{\mlloc}_c\,
\]
for every natural number~$n$.
These maps satisfy axioms analogous to {\bf (T1) -- (T6)}, which can be found in \cite{H}, but for completeness of our review we spell them out here. 

\begin{enumerate}[{\bf (TT 1)}] 
\item{\bf Starting element}:  $\Tb_0=1$ and $\Tb_1=\id$.
\end{enumerate}

The analog of microcausal functionals is a bit subtle.
It is inspired by the concrete construction,
which uses Feynman diagrammatic methods.
Thus, following \cite{BFK96,H} we will formulate the microlocal spectrum condition using diagrams, as follows.  
Let $\Gamma$ be a finite graph embedded in $X$ where the vertices $V(\Gamma)$ are points $x_1,\dots, x_n\in X$ and where the edges $e$ are oriented null-geodesic curves in~$X$. 
Each such null geodesic is equipped with a coparallel, cotangent covectorfield $p_e$. 
If $e$ is an edge in $\Gamma$ connecting the points $x_i$ and $x_j$ with $i < j$, then its source is $s(e) = i$ and its target is $t(e)= j$. 
We require that $p_e$ is future (past, respectively) directed if $x_{s(e)} \notin J^+(x_{t(e)})$ (or not in $J^+(x_{t(e)})$, respectively). 
We define a subset $\widetilde{\Xi}_n$ of $T^*X^n$ by
\be\label{def: CT}
\widetilde{\Xi}_n \doteq 
\left\{
(x_1, k_1;\dots ; x_n, k_n) \in \dot{T}^*X^n \, \Big|\, 
\substack{
\exists\, \Gamma\ \textrm{with}\ V(\Gamma)=\{x_1,\dots, x_n\}\ \textrm{such that} \\ 
 \forall i\,\ k_i = \sum_{e:s(e)=i} p_e  - \sum_{e:t(e)=i} p_e
}
\right\}.
\ee
(Hollands uses the notation $C_T$ for this subspace.)
This subspace arises naturally by looking at the integrands of Feynman diagrams, where the edges have the Hadamard state associated to each edge.
The axiom ends up being guaranteed by construction.
(An alternative characterization would be nice.)

We now formulate objects of key interest.

\begin{df}
A distributional $n$-fold multilocal differential form $\alpha \in \overline{\lloc}^{[n]}_c$ is {\em microcausal} if $\mathrm{WF}(\alpha)\subset \widetilde{\Xi}_n$.
\end{df}

The star product on microcausal differential forms is defined analogously to the one on microcausal functionals. 
For $\alpha,\beta\in \overline{\mlloc}_{\mu c}$, we define
\begin{equation*}
\alpha \star \beta\doteq m\circ \exp({i\hbar D_W})(\alpha\otimes \beta),
\end{equation*}
where $m$ is the pointwise --- in field space $\Ecal$ --- product of Lagrangian forms. 

\begin{enumerate}[{\bf (TT 1)}] 
\addtocounter{enumi}{1}
\item{\bf Microlocal spectrum condition}: We require $\Tb_n: \lloc_c^{[n]} \to \overline{\mlloc}_{\mu c}$ for all~$n$.
\end{enumerate}

\begin{enumerate}[{\bf (TT 1)}] 
\addtocounter{enumi}{2}
\item{\bf Causal factorization property}: There is a commuting diagram
\[
\begin{tikzcd}[column sep=large] \lloc_c^{[k]}\preceq \lloc_c^{[n-k]}\arrow[r,"\Tb_{n-k}\times\Tb_{k}"]\arrow[d,"\cdot"] &\overline{\mlloc}_{\mu c} \times \overline{\mlloc}_{\mu c}\arrow[d,"\star"]\\ \lloc^{[n]}\arrow[r,"\Tb_n"] & \overline{\mlloc}_{\mu c}
\end{tikzcd}
\]
where 
\[
\lloc_c^{[k]}\preceq \lloc_c^{[n-k]}\equiv \left\{(\alpha,\beta)\in \lloc_c^{[k]}\times \lloc_c^{[n-k]} |\, \alpha \preceq \beta\right\}\,,
\]
and $\alpha \preceq \beta$ means that $\supp \alpha \preceq \supp \beta$.

\item {\bf Symmetry}: 
For all $n\in\NN$, the map $\Tb_n$ intertwines the $S_n$-actions on the (co)domains arising from the permutation action on~$X^n$.
(In other words, $\Tb_n$ is $S_n$-equivariant.)

\item {\bf Field independence}: There is a commuting diagram
\[
\begin{tikzcd}[column sep=large] \lloc_c^{[n]}\arrow[r,"(\bullet )^{\psi}"]\arrow[d,"\Tb_n"] &\lloc^{[n]}_c \arrow[d,"\Tb_n"]\\ 
\overline{\mlloc}_{\mu c}\arrow[r,"(\bullet )^{\psi}"] & \overline{\mlloc}_{\mu c}
\end{tikzcd}
\]
for all $n$.

\item {\bf $\ph$-Locality}: For all $n$ and $N$, $\Tb_{n}(F)=\Tb_{n}(F^{[N]})+\Ocal(\hbar^{N+1})$ for all~$F$.

\item {\bf Covariance} $\Tb_n$ is a natural transformation from the functor $\lloc_c^{[n]}$ to $\overline{\mlloc}_{\mu c}$ for all~$n$.
\end{enumerate} 

In order to use the maps $\Tb_n$ to construct the maps $\TT_n$, we need to make sure that $\Tb_n$s commute with derivatives. 
This is formulated using the condition below, first proposed in~\cite{DF05,H}.
(The relation with Lie derivative of vector fields is via Cartan's formula.)

\begin{enumerate}[{\bf (TT 1)}] 
\addtocounter{enumi}{8}
\item{\bf Action Ward Identity}
For every $n$, the map $\Tb_n$ is a cochain map.
In formulas,
\[
\d_\mlloc(\Tb_n(\alpha))= \Tb_n(\d_\mlloc \alpha)\,
\]
for every $\alpha \in \lloc_c^{[n]}$.
\end{enumerate} 

Given a family of maps $\Tb_n$ satisfying \textbf{(TT 1) -- (TT 9)}, we obtain maps $\TT_n$ satisfying \textbf{(T 1) -- (T 7)} by the ansatz 
\begin{equation}
\label{ansatz}
{\Tcal}_{n}(F_{\alpha})\doteq \int_{X^n} \Tb_n(\alpha)\,
\end{equation}
for the multilocal functional
\[
F_\alpha(\phi_1, \ldots, \phi_m)= \int_{X^n} \alpha(\jet_{k_1}\phi_1(x_1), \ldots, \jet_{k_m}\phi_m(x_m))
\]
associated to an $n$-fold Lagrangian density~$\alpha$.

\subsubsection{The existence theorem}

For the construction of the maps $\Tb_n$, 
we use the by-now standard arguments of \cite{H,BDF} (see also \cite{Due19} for review). 
Firstly, properties \textbf{(TT 5)} and \textbf{(TT 6)} let us obtain a Taylor series expansion of $\Tb_n(\alpha)$, where at each order in $\hbar$ only finitely many term in this expansion contribute. 
We have that
\be\label{eq: TT from t}
\Tb_n(\alpha)(\jet_\infty \phi_1,\dots,\jet_\infty \phi_n)=\sum_{(\al)}\Tb_n(\alpha_{(1)})(0)
\alpha_{(2)}(\jet_\infty \phi_1,\dots,\jet_\infty \phi_n)\,,
\ee
using the coproduct of Section~\ref{fbrws cprdct ml} to expand $\alpha$ into its Taylor series around the zero section. 
This literature calls this the \emph{Wick expansion}, 
and it relates the problem of constructing the time-ordered products to constructing a  de Rham current (or distributional de Rham form)
\[
\Tb_n(\alpha)(0)\equiv t_n(\alpha)
\]
for every fiberwise analytic Lagrangians~$\al$. 
In other words, we want to have a degree-preserving map
\[
t_n: \lloc_c^{[n]}(X^n) \to \overline{\Omega}^*_{X^n,c}(X^n)
\]
for every~$n$.
These maps should satisfy axioms compatible with the axioms for the maps~$\Tb_n$.

Note that axioms ({\bf TT~1}) and ({\bf TT~3}) fix each $t_n(\alpha)$ everywhere outside the total diagonal in~$X^n$.

Axiom ({\bf TT~2}) is guaranteed by the following.

\begin{enumerate}[{\bf (t 1)}]
\item{\bf Microlocal spectrum condition}:  
For every $\alpha$, 
$$\mathrm{WF}(t_n(\alpha))\big|_{X_\Delta} \subset 
\left\{x,k_1;\dots;x,k_n \, \Big| \,\sum_{i=1}^n k_i=0\right\}$$ 
where $X_\Delta$ denotes the small (or total) diagonal in~$X^n$.
 \end{enumerate}
 
Axioms ({\bf TT~4}) is guaranteed by the following.  

\begin{enumerate}[{\bf (t 1)}]
\addtocounter{enumi}{1}
\item{\bf Symmetry}: For every $n$,
the map $t_n$ is $S_n$-equivariant.
\end{enumerate}

Axioms ({\bf TT~5}) and ({\bf TT~6}) do not impose any properties on the $t_n$. 
Axiom ({\bf TT~7}) motivates the following.

\begin{enumerate}[{\bf (t 1)}]
\addtocounter{enumi}{2}
\item{\bf Covariance}: For every $n$, the map $t_n$ is map of cosheaves. 
\end{enumerate}

In the construction, this feature is guaranteed because the map $t_n$ depends locally and covariantly on the metric on~$X$.

Axiom ({\bf TT~9}) is guaranteed by the following.

\begin{enumerate}[{\bf (t 1)}]
\addtocounter{enumi}{3}
\item{\bf Action Ward Identity}: 
The map $t_n$ is a cochain map. That is, $\d(t_n(\alpha))=t_n(\d_{\mlloc}(\alpha))$ for every~$\alpha$.
 \end{enumerate}

There are three more axioms in section 3.3 of~\cite{HW}, with strong physical and technical motivation, but they are not relevant for our work here so we mention them without explaining them.
\begin{enumerate}[{\bf (t 1)}]
\setcounter{enumi}{4}
\item{\bf Almost homogeneous scaling}.
\item{\bf Smoothness}.
\item{\bf Analyticity}.
\end{enumerate}
The first guarantees that each ${t}_{n}(\alpha)$ has desirable scaling properties. 
The other two characterize the dependence on the metric. 

The key result is the following.

\begin{thm}[after \cite{H}]
There exist a family of maps $t_{n}$ satisfying the axioms \em{({\bf t~1})-({\bf t~7})}, and they determine maps $\Tb_n$ that satisfy axioms ({\bf TT~6}) -- ({\bf TT~7}) and ({\bf TT~9}) by means of formula~\eqref{eq: TT from t}. 
\end{thm}

\begin{proof}
For the proof, see \cite[section 3.3]{H}, as well as \cite[section 4.1]{BDF}, \cite[section 4]{BF0} and \cite[section 3.3]{H}. 
(A more robust and mathematically cleaner method for constructing such distributional extensions is offered in~\cite{Viet}.) 
The main idea is that the formula~\eqref{eq: TT from t} guarantees that all we need to do is to construct the ``coefficients'' $t_n(\alpha)$, 
which are just distributional de Rham forms, not functions on~$\Ecal$. 
The process is inductive on $n$.
Crucially, {\bf (TT 3)} means that outside the small diagonal, 
the form $t_n(\alpha)$ is completely determined by the $t_k(\alpha)$, for $k<n$. 
We also know that the other desired properties are fulfilled everywhere outside the diagonal (because the $t_k(\alpha)$, $k<n$ and corresponding $\Tb_k$s satisfy all the axioms, by hypothesis). 
Hence at each step of the induction one needs to show that distributional extensions to the small diagonal exist. 
These are not unique and some of the remaining freedom is used to fulfill all the remaining axioms for $t_n$s and~$\Tb_n$s. 
\end{proof}

One obtains $\TT_n$s from $\Tb_n$s using the ansatz~\eqref{ansatz}. Note that the theorem guarantees existence, but not uniqueness.


\subsection{The main theorem of renormalization}

There is some ambiguity in the definition of time-ordered products of local functionals,
so there are many possible time-ordered products.
This variety of choices is controlled by the St\"uckelberg-Petermann renormalization group \cite{SP53,SP82} or rather its modern generalization~\cite{BDF,HW02,H}.
The essential point, stated below as the {\em main theorem}, is that the action of the group is transitive and thus relates all choices of time-ordered product.

\begin{df}
A {\em formal map} $\Zcal:\loc[[\lambda]]\rightarrow \loc [[\lambda]]$ is a formal power series (in $\lambda$) of the form $\Zcal(F)=\sum_{n=0}^\infty \lambda^n Z_n(F^{\otimes n})$, and such maps compose as power series.
A {\em formal diffeomorphism} is a formal map that has an inverse among formal maps.

The {\em renormalization group} $\Rcal$ is the group of formal diffeomorphisms $\Zcal:\loc_c[[\lambda]]\rightarrow \loc_c[[\lambda]]$ that satisfy
\begin{enumerate}[{\bf Z1}]
	\item {\bf Identity preservation}: $\Zcal(0)=0$.\label{Z:ident}
	\item {\bf Causality}:\label{Z:loc} $\Zcal$ 
		satisfies the Hammerstein property, i.e. $F_1\preceq F_2$ implies that
		\[
		\Zcal(F_1+F+F_2)=	\Zcal(F_1+F)-\Zcal(F)+	\Zcal(F_2+F)\,.
		\]
	
	\item {\bf Symmetry}: The maps $Z_n$ are graded symmetric under permutations of arguments. 
		
	\item {\bf Field independence}: For any field $\psi\in\Ecal$, we have
 \[
\Zcal(F^\psi)=\Zcal(F)^\psi\,.
 \]	

 	\item {\bf $\ph$-Locality}: $\Zcal_{n}(F_1,\ldots,F_n)=\Zcal_{n}(F_1^{[N]},\ldots,F_n^{[N]})+\Ocal(\hbar^{N+1})$, where  $F_i^{[N]}$ is the Taylor series expansion of the functional $F_i$ up to the $N$-th order.
	\item {\bf Covariance} $\Zcal$ is a natural transformation between the functors $\loc_c$ and $\loc_c$.
    \item {\bf Unitarity: }	$\overline{\Zcal}(-V)+\Zcal(V)=0$.
	\end{enumerate}
\end{df}

Note the close relation of this list with the list for time-ordered products.
To see the relevance of  this notion of the renormalization group in our context,
we introduce the notion of a formal S-matrix.

\begin{df}
Let $\TT$ be some time-ordering operator construction~$\TT$ (as given in Definition~\ref{df: T map}).
For $F\in \loc_c$,
its {\em formal S-matrix} $\Scal_\TT(F)$~is
\[
\Scal_\TT(\lambda F) \doteq \TT e^{i\lambda F/\hbar} 
\]
where $\lambda$ is a formal parameter and where we apply the time-ordering to the formal power series of the exponential.
This S-matrix thus takes values in~$\overline{\mloc}_{\mu c}[[\lambda]]$.
\end{df}

The ambiguity of defining time-ordered products of local functionals translates to ambiguity in defining an S-matrix for a given $F$. 
This ambiguity, however, is completely characterized by the action of the renormalization group, as shown below.

\begin{thm}[Main theorem of renormalization]
\label{thm:renorm}

The renormalization group $\Rcal$ acts on $\loc[[\lambda]]$
and it induces an action on formal S-matrices by {\em precomposition}:
for any $\Zcal\in\Rcal$ and for any time-ordered product $\TT$,
the composite $\Scal_{\TT} \circ \Zcal$ is also a local S-matrix.
In other words, there exists a time-ordered product $\tilde{\Tcal}$ such that 
\[
\Scal_{\TT} \circ \Zcal(F)=\Scal_{\tilde{\Tcal}}.
\]
Moreover, this induced action on local S-matrices is transitive.
\end{thm}

Instead of working just with the space of local functionals, one can also start with an arbitrary additive group equipped with \textit{causality relation}, as shown in \cite{Rej19}. 
This approach applies to our local differential forms $\lloc_c$, 
where the causality relation is induced by the $\preceq$ relation for their supports. However, we postpone this discussion to our later work; 
see the remark below.

\begin{proof}[Proof outline]
One direction is straightforward. 
A short calculation shows that for $\Zcal\in\Rcal$, the composite $\Scal\circ \Zcal$  is a local S-matrix. 
To go the other way, the non-obvious step is to show ${\bf Z2}$. 
One does this by construction, and 
the construction has been spelled out in sufficient generality in the proof of Theorem~4.7 in \cite{Rej19}. 
The remaining properties follow by direct inspection, since the $\Zcal_{n}$s are constructed inductively as differences of appropriate $n$-fold time-ordered products, 
so they inherit all their properties.
\end{proof}

\begin{rem}
In the proof one shows that $\Zcal_n$ is local by showing that it is supported on the small diagonal. 
When working with distributional multilocal differential forms, 
one cannot identify a multilocal form supported on the diagonal with a local one.
However, if one works on the Ran space of $X$ (i.e., the space of finite subsets of $X$, which identifies $X$ with the small diagonal in every $X^n$),
this happens naturally.
Hence, in order to generalize the above main theorem to $\lloc_c$, 
it seems natural to work on the Ran space. 
We hope to address this in our future work.
\end{rem}


\section{The classical observables}
\label{sec:classobs}

We return now from renormalization to examining the observables of a classical field theory.
Everything so far is essentially kinematic; 
we have not imposed the equations of motion to obtain classical observables.
To encode these dynamics,
we use a differential.
This differential is a vector field on the graded space $\Ecal$ of field configurations,
albeit of cohomological degree one,
and its kernel consists of functions on the fixed points of the vector field.

In the BV formalism for classical field theory, this differential appears as a kind of Hamiltonian vector field. 
In particular, the local functionals possess a shifted Poisson bracket --- known as the antibracket --- so that the differential is the Hamiltonian vector field $\{S,-\}$, namely the antibracket with the action functional~$S$.
We review the {\em shifted} symplectic geometry relevant to our context
before turning to constructing a differential on our version of multilocal 
differential forms (and multilocal functionals),
which allows us to describe our factorization algebra of classical observables.
This construction is borrowed from~\cite{BarHen, H, Get21}.

Until the final subsection~\ref{analyticity of obscl},
we do not require fiberwise analyticitiy.

\subsection{A shifted symplectic structure on fields}\label{sec: symplectic}

In the BV formalism, we require the following structure on the field configurations.

\begin{df}
Let $E \to X$ be a graded vector bundle on a manifold~$X$.
A {\em -1-shifted symplectic bundle pairing} $\omega$ is a map of graded vector bundles
\[
\omega: E \times_X E \to \Dens_X
\]
that is fiberwise nondegenerate, has degree -1, and is skew-symmetric in the sense that 
\[
\omega( e, e') = -(-1)^{(|e|+1)(|e'|+1)} \omega( e', e)
\]
for any point $x \in X$ and for any elements $e, e' \in E_x$, where $E_x$ denotes the fiber of the bundle $E$ at~$x$.

Such a pairing determines a {\em -1-shifted local symplectic pairing} $\omega_\Ecal$ on the graded space $\Ecal$ of field configurations 
\[ 
\omega_\Ecal: \Ecal \times \Ecal \to \cDens
\]
where the density $\omega_\Ecal(\phi, \phi')$ satisfies
\[
\omega_\Ecal( \phi, \phi')(x) = \omega(\phi(x), \phi'(x))
\]
for any fields $\phi, \phi' \in \Ecal$.
\end{df}

\def\ev{{\rm ev}}

\begin{rem}
In typical examples of the BV formalism, there is a graded vector bundle $F \to X$ whose degree 0 sections provide the fundamental fields for the theory (e.g., connection 1-forms) and whose degree -1 sections provide the ghosts, if those are needed.
The full graded bundle $E \to X$ arises by adjoining ``antifields'' so $E^0 = F^0$, $E^{-1} = F^{-1}$, but $E^{1} = (F^0)^* \otimes \Dens$ and $E^2 = (F^{-1})^* \otimes \Dens$.
By construction, there is then a canonical $\omega$ on the vector bundle $E$ by leveraging the evaluation pairing. 
For instance, we have the map
\[
E^0 \times_X E^1 \cong F^0 \times_X (F^0)^* \otimes \Dens \xto{\ev_{F^0} \otimes \id_\Dens} \Dens
\]
and similarly between $E^{-1}$ and $E^{2}$ (i.e., the ghosts and their antifields). As 
\end{rem}

Consider the vector bundle 
\[
E^! = \Hom(E, \Dens) \cong E^* \otimes \Dens
\]
where $E^*$ denotes the dual bundle to $E$.
Let $\Ecal^!$ denote smooth sections of $E^!$,
and similarly let $\Ecal^!_c$ denote compactly supported smooth sections of~$E^!$.

Note that our pairing $\omega$ on $E$ induces an isomorphism $\omega^\flat: E \xto{\cong} E^!$.
We thus obtain an isomorphism $\omega^\flat_\Ecal: \Ecal \xto{\cong} \Ecal^!$ induced by that isomorphism of vector bundles.
Let $\omega^\sharp: E^! \to E$ denote the inverse to $\omega^\flat$, 
and likewise we write $\omega^\sharp_{\Ecal}: \Ecal^! \to \Ecal$ for the induced isomorphism.
We thus obtain a pairing
\[
\omega_{\Ecal^!}: \Ecal^! \times \Ecal^! \to \cDens
\]
by
\[
\omega_{\Ecal^!} = \omega_\Ecal \circ (\omega^\sharp_{\Ecal} \times \omega^\sharp_{\Ecal} ).
\]
Note that this pairing is canonical but it has two features that might look strange in comparison to the traditional construction.
First, the continuous linear dual $\Ecal'$ of $\Ecal$ consists of {\it distributional} sections of $E^!$, 
whereas this pairing is only defined on $\Ecal^!$, the {\it smooth} sections.
Second, the pairing takes values in $\cDens$.
It is, however, the natural codomain for our setting, 
as we now explain by examining the differentials of Lagrangian top forms.
(Recall that we implicitly assume $X$ is orientable, so that densities are top forms, and leave it to the reader to make the small modifications in the unoriented case.)

\subsection{A reminder on generalised Lagrangians}
\label{sec:genL}

In pAQFT one uses the notion of a  {\em generalised Lagrangian}, as introduced in~\cite{BDF}. 
The issue is that since we work on globally hyperbolic and hence non-compact manifolds,  the integral of a non-compactly supported Lagrangian density would diverge. 
The Lagrangian densities used in action functionals are not compactly supported, 
but they can be made into such by multiplying with a test function (i.e. compactly supported function on $X$). 
In general, the dependence on test functions can be more complicated and this freedom is actually needed in order to treat physically interesting examples like Yang-Mills (see Section~\ref{sec: examples}). 
With this in mind, we introduce the following definition.

\begin{df}\label{Lagr}
Let $\euD\doteq \cinfty_c(X,\RR^N)$ be the space of test functions (the choice of $N$ depends on the model). A {\em generalised Lagrangian} on $X$ is a map $L$ from $\euD$ to local functionals on $\Ecal$ that is
\begin{itemize}
\item {\em additive}: $L(f+g+h)=L(f+g)-L(g)+L(g+h)$ for all $f,g,h\in\euD$ with $\supp\,f\cap\supp\,h=\varnothing$, and
\item {\em preserves support}: $\supp(L(f))\subseteq \supp(f)$ 
\end{itemize}
\end{df}

Any Lagrangian density $\Lcal$ produces a generalized Lagrangian~$L$,
even if this Lagrangian density $\Lcal$ has {\em noncompact} support.
For example, let $\Lcal(\phi)=\alpha(\jet \phi) \mu$ where $\mu$ is a volume form on $X$ and $\alpha$ is an order $k$ Lagrangian. 
Then we define
\[
L(f)(\phi)=\int_X f\Lcal(\phi)= \int_X f\alpha(\jet_k \phi) \mu\,
\] 
where $\phi$ is an arbitrary field and $f\in\Ci_c(X)$.
Another useful natural option is to multiply the fields by $f$: define
\[
L(f)(\phi)=\int_X \Lcal(f\phi)= \int_X \alpha(\jet_k (f\phi)) \mu
\] 
if $\alpha$ is at least linear in fields.


A generalized Lagrangian $L$  encodes the dynamics of a theory, 
by means of it {\em generalized action},  
often called the {\em action functional} of the theory. 
The action is an equivalence class of generalized Lagrangians
under the relation $L_1\sim L_2$~if
\be\label{equ}
\supp (L_{1}(f)-L_{2}(f))\subset\supp\, \d f\,, 
\ee
for any $f\in\euD$.
See \cite{BDF} for further development of this approach.

\subsection{The cohomological vector field and dynamics}
\label{sec: dynamics}

Via the construction \eqref{eqn Fp}, a local functional $F = \Upsilon^n(\alpha, \mu)$ is an element of ${\rm Maps}(\Ecal, \CC)$,
where 
\[
F(\phi) = \int_X \alpha(\jet_k(\phi))\mu.
\]
This map is actually {\em smooth} and so it has a differential $\delta F$.
(We use $\delta F$ rather than $\d F$ as $\d$ will be reserved for the exterior derivative on the spacetime manifold~$X$.)
As $\Ecal$ is a vector space, its tangent bundle admits a natural trivialization by translation, so for any $\phi \in \Ecal$,  there is an isomorphism $T_\phi \Ecal \cong \Ecal$. 
Similarly, 
\[
T^*_\phi \Ecal \cong \Hom_{cts}(\Ecal, \CC) = \Ecal',
\]
and $\delta F|_\phi$ is a ``covector'' in $T^*_\phi \Ecal$.
Thus, given a tangent vector $\psi \in T_\phi \Ecal \cong \Ecal$,
it pairs with the differential $\delta F|_\phi \in T^*_\phi \Ecal$ to give the directional derivative
\[
(\psi, \delta F|_\phi)= D_\psi F(\phi) = \lim_{\epsilon \to 0} \frac{F(\phi + \epsilon \psi) - F(\phi)}{\epsilon},
\]
as usual.
Varying over all input fields, the differential $\delta F$ can also be seen as a smooth map from $\Ecal$ to~$\Ecal' = \Hom_{cts}(\Ecal, \CC)$.

There is a crucial feature of such local functionals. 

\begin{lemma}[\cite{BDGR}]
For any local functional $F$ and for any field $\phi \in \Ecal$, the derivative $\delta F|_\phi$ is a {\em smooth} functional of the tangent direction $\psi$. 
That is, $\delta F|_\phi$ is an element of~$\Ecal^!_c$.
\end{lemma}

To unpack this statement, if $\overline{\Ecal}$ denotes the distributional completion of $\Ecal$ so that $\Ecal \subset \overline{\Ecal}$, then
\[
\delta F|_\phi \in \Hom_{cts}( \overline{\Ecal}, \CC) \subset \Hom_{cts}(\Ecal, \CC) = \Ecal'.
\]
However, $\Ecal^!_c$ is precisely~$\Hom_{cts}(\overline{\Ecal}, \CC)$.
(Given any distributional section of $E$, it pairs with any compactly supported smooth section of $E^!$ to produce a number,
and conversely any continuous linear functional on such distributional sections is of this form.)

Concretely, this lemma reflects that for a functional of the form
\[
F(\phi) = \int_X D_1 \phi \cdots D_k \phi \,\mu
\]
where the $D_j$ are differential operators and $\mu$ is a density,
its variation
\[
D_\psi F(\phi) = \sum_{j = 1}^k \int_X D_1 \phi \cdots D_j \psi \cdots D_k \phi \,\mu
\]
is well-defined even if $\psi$ is a distribution.
In this sense, the derivative $\delta F$ --- computed by variations against smooth fields --- can be pulled back from a linear functional on distributional fields.

Thanks to this lemma, the differential $\delta F$ of a local functional $F$ is a smooth map from fields $\Ecal$ to $\Ecal^!_c$, 
as one can pair $\delta F(\phi)$ with an arbitrary distributional section $\psi$ of~$E$.
By postcomposing with the isomorphism $\omega^\flat_\Ecal: \Ecal^! \to\Ecal$, 
we obtain a smooth map from fields $\Ecal$ to $\Ecal$ itself, now viewed as the tangent bundle.
We can view this map as a vector field $\Qcal_S$ on the graded manifold~$\Ecal$.

We now turn to the (noncompactly supported) Lagrangian density $\Lcal$ that encodes the dynamics of the theory.
It determines a generalized action, as discussed in section~\ref{sec:genL},
which we will call the {\em action functional} $S$ of the theory.
The dynamics of the theory is given in terms of a 1-form $\delta S$ on $\Ecal$, defined by
\[
(\psi, \delta S|_\phi):=(\psi, \delta L(f)|_\phi)\,,
\] 
where $\psi \in \overline{\Ecal}_c$ and on the support $\supp \psi$, $f\equiv 1$. 
It follows from locality that $\delta S$ defined this way does not depend on the choice of the cutoff $f$, 
and so it is a 1-form on $\Ecal$ and defines a vector field $\Qcal_S$, as explained above.

\begin{df}
A {\em classical BV theory} consists of a graded space $\Ecal$ of fields equipped with the -1-shifted local symplectic pairing $\omega_\Ecal$ and an  action functional $S$ such that the induced degree 1 vector field $\Qcal_S$ 
satisfies 
\begin{equation}
\label{eqn:cme}
    \Qcal^2_S = 0,
\end{equation}
a condition known as the {\em classical master equation (CME)}. 
\end{df}

As a vector field $\Qcal_S$ acts as a derivation on the ring of smooth functions on $\Ecal$:
for $F$ a smooth function,
\[
\Qcal_S(F) = L_{\Qcal_S}(F) = \iota_{\Qcal_S}(\delta F),
\]
where $L_{\Qcal_S}$ denotes the Lie derivative and $\iota_{\Qcal_S}$ denotes contraction.
We are interested, however, in local functionals and eventually local differential forms.
As one might hope, $\Qcal_S$ preserves the subspace of local functionals.

\begin{lemma}
For a local functional $F \in \loc_c$, the functional $\Qcal_S(F)$ is local.
Thus, $(\loc_c(X), \Qcal_S)$ is a cochain complex.
\end{lemma}

\begin{proof}
Unraveling the definitions, one finds that
\[
\Qcal_S(F)(\phi) = \int_X \omega_{\Ecal^!}(\delta S|_\phi, \delta F|_\phi).
\]
This formula is manifestly a local functional.
\end{proof}

\begin{rem}
The reader already familiar with the BV formalism will recognize that $\Qcal_S(F) = \{S,F\}$, where $\{-,-\}$ denotes the {\em antibracket}.
Indeed, the antibracket is defined on local functionals by
\begin{equation}
\label{lag 0 pairing}
\{ F, G \}(\phi) = \int_X \omega_{\Ecal^!}(\delta F|_\phi, \delta G|_\phi),
\end{equation}
which is precisely the formula we used in the preceding proof.
\end{rem}

\subsection{Equations of motion}

Let $L^0(f)=\int_X f\Lcal^0$ be the term in the Lagrangian $L$ with no antifields. Let $S^0$ denote the corresponding generalized action. The equation of motion of the theory is the condition
\be\label{EOM}
\delta S^0(\phi)\equiv 0\,,
\ee
and fields $\phi\in\Ecal$ satisfying this condition are \textit{solutions to the equation of motion}. 

The remaining term $\theta$ of the generalized Lagrangian encodes symmetries and potentially higher order data (e.g. depenedencies between symmetries). Physically, we are interested in the space of \textit{gauge-invariant on-shell functionals}, i.e. functionals on the space of solutions to \eqref{EOM} (i.e. the zero locus of $\delta S^0$) that are invariant under the symmetries. 

\begin{rem}
One can express \eqref{EOM} as a condition upon the jets of $\phi$, i.e., a system of partial differential equations, as follows. 
Every local functional is defined by a Lagrangian density,
but many Lagrangian densities determine the same local functional.
Indeed, if two Lagrangian densities differ by an exact term with respect to the differential $\d_\lloc$ on $\lloc$,
then they determine the same local functional.
This property has consequences for the differential $\delta F$ of a local functional $F$.
In particular, if $F(\phi) = \int_X \Lcal(\phi)$, then there exists a Lagrangian density $\delta^{r/\ell}_{EL} \Lcal$ (where the superscript denotes ``right/left'' and the subscript denotes ``Euler-Lagrange'')
such that the directional derivative satisfies
\[
D_\psi F(\phi) = \int_X \delta^{\ell}_{EL} \Lcal(\phi) \, \psi = \int_X \psi \, \delta^{r}_{EL} \Lcal(\phi).
\]
There are two Euler-Lagrange derivatives to account for signs that arise from fields with odd cohomological degree.
These Euler-Lagrange derivatives exist because one can use integration by parts to move any derivatives off the term $\psi$, which is a tangent vector on the space of fields, and onto the rest of the density.
For simplicity, we work exclusively with $\delta^{\ell}_{EL}$ from hereon and use $\delta^{\ell}_{EL}$, dropping the superscript.
	
We can now express \eqref{EOM} as:
\[
\delta^{\ell}_{EL}\Lcal^0=0\,.
\]
These are precisely the usual equations of motion for the field theory.
\end{rem}

\subsection{The factorization algebra of classical observables}

Another important property of $\Qcal_S$ is that it preserves supports of functionals, 
a property needed to obtain a factorization algebra.
Apply the notion of support of an observable from section~\ref{sec on obs as precosheaf}: a local functional $F \in \loc_c(X) \subset \cinfty(\Ecal(X))$ has support in $K \subset X$ if $K$ is the smallest closed set for which the injection $\loc_c(K) \hookrightarrow \loc_c(X)$ has $F$ in its image.

\begin{lemma}
The cohomological vector field $\Qcal_S$ preserves support of local functionals.
\end{lemma}

\begin{proof}
If a local functional $F$ has support in $K \subset X$, then $\delta F$ also has support in $K$,
where again we mean $\delta F(\phi)$ has support contained in $K$ for any field $\phi$.
Hence ${\rm supp}(\delta F) \subset {\rm supp}(F)$.
As $\omega_{\Ecal^!}$ arises from a vector bundle pairing, 
it also respects supports.

Let $S$ be the action functional and let $F$ be a local functional. 
Then 
\[
\Qcal_S(F) = L_{\Qcal_S}(F) = \iota_{\Qcal_S}(\delta F),
\]
where $L_{\Qcal_S}$ denotes the Lie derivative and $\iota_{\Qcal_S}$ denotes contraction.
We thus see that
\[
{\rm supp}(\Qcal_S(F)) \subset {\rm supp}(\delta S) \cap {\rm supp}(\delta F) \subset {\rm supp}(S) \cap {\rm supp}(F)   
\]
as claimed.
\end{proof}

As support is preserved, we get a (strict) cosheaf of cochain complexes of local functionals.

\begin{cor}
The functor from $\Opens(X)$ to $\Ch$ assigning the cochain complex $(\loc_c(U), \Qcal_S)$ to each open set~$U$ is a cosheaf of cochain complexes.
\end{cor}

\begin{rem}
This target category is acceptable, but these vector spaces have additional structure.
For instance, one might view them as topological or bornological vector spaces (or some close cousin).
Tracking this additional structure is useful in many contexts but we will not focus on it here.
\end{rem}

This functor is not a homotopy cosheaf,
as $\loc_c$ is not a homotopy cosheaf.
Indeed, one reason to replace $\loc_c$ by $\lloc_c$ is to obtain a construction that satisfies a local-to-global principle that is compatible with the dg structure (i.e., works up to quasi-isomorphism).

The differential $\Qcal_S$ extends naturally to multilocal functionals,
as shown in \cite{FR};
those arguments extend immediately to our fattened multilocal functionals, as done, for example in section 2.7 of \cite{BFRej13}.
We record the result here.

\begin{lemma}
If the action functional $S$ satisfies the classical master equation, 
then $\Qcal_S$ extends to a differential on $\mloc_c$, 
and this differential is a derivation with respect to the commutative product on~$\mloc_c$.
\end{lemma}

We thus know that $(\mloc_c,\Qcal_S)$ is a prefactorization algebra with values in $\Ch$.
(In fact, it takes values in dg commutative algebras,
as it assigns a dg commutative algebra to each open set and the structure maps all respect the dg commutative algebra structures.)
On the other hand, by Proposition~\ref{prop: mloc is fact alg}, we know that $\mloc_c$ (with zero differential) is a strict factorization algebra with values in $\Ch$.
Note, however, that a strict colimit of a diagram of cochain complexes is computed degreewise (i.e., in each degree, take the colimit of vector spaces and similarly for the differential).
Hence we immediately deduce the following.

\begin{prop}
\label{prop: obscl is a fact alg}
Moreover, there is a factorization algebra on $X$ with values in $\Ch$ that assigns the cochain complex $(\mloc_c(U), \Qcal_S)$ to each open set~$U$.
\end{prop}

This object is of central interest for this paper.

\begin{df}
We call $(\mloc_c, \Qcal_S)$ the {\em algebra of classical observables} given by multilocal functionals.
\end{df}

This construction recovers the net of dg Poisson algebras already constructed in the pAQFT version of the BV formalism.
Consider the subcategory ${\bf Caus}(X) \subset \Opens(X)$ of causally convex, relatively compact opens. 
If one restricts $(\mloc_c, \Qcal_S)$ to ${\bf Caus}(X)$, 
obtains the net of classical observables.
(In \cite{FR} this net is denoted $\frak{BV}$.)
See Section~\ref{review of GR} for more context.

\subsection{Extending the differential to dg local functionals}

We want to lift $\Qcal_S$ from $\loc_c$ to $\lloc_c$, and later to dg multilocal functionals.

A Lagrangian $p$-form $F = \Upsilon^p(\alpha, \mu)$ is an element of ${\rm Maps}(\Ecal, \Omega^p)$.
It is actually a {\em smooth} map and so it has a differential
\[
T(F) : T \Ecal \to T \Omega^p, 
\]
but as these are both vector spaces, $T(f)$ defines a map 
\[
\begin{array}{cccc}
T(F): & \Ecal \times \Ecal & \to & \Omega^p \times \Omega^p\\
& (\phi, \psi) & \mapsto & \left(F(\phi), \frac{\delta F}{\delta \psi}(\phi) \right)
\end{array}.
\]
Here $F(\phi) = \alpha(\jet_k(\phi))\mu$ while the directional derivative
\[
D_\psi F(\phi) = \lim_{\epsilon \to 0} \frac{F(\phi + \epsilon \psi) - F(\phi)}{\epsilon}
\]
as usual.
Let $\delta F: \Ecal \times \Ecal \to \Omega^p$ denote the projection of $T(F)$ onto the second component of the codomain,
i.e., 
\[
\delta F(\phi, \psi) = D_\psi F(\phi),
\]
which can also be seen as a smooth map from $\Ecal$ to~$\Hom_{cts}(\Ecal, \Omega^p)$,
i.e., as a smooth 1-form on $\Ecal$ with values in $\Omega^p$.
We will write $\delta F|_\phi$ to denote the ``covector'' in $\Hom_{cts}(\Ecal, \Omega^p)$,
which can be viewed as an $\Omega^p$-valued functional on the tangent space~$T_\phi \Ecal \cong \Ecal$.
There is a natural pairing between 1-forms and vector fields on any manifold, even $\Ecal$,
which we now put to good use.

\begin{df}
Define an operator $\delta_S$ on Lagrangian $p$-forms by the formula
\[
\delta_S(F) = \iota_{\Qcal_S}(\delta F),
\]
where $\iota$ denotes contraction of a vector field with a 1-form and the term $\Qcal_S$ means the cohomological vector field on $\Ecal$ defined by an action function.
\end{df}

This operator is, in fact, an endomorphism.

\begin{lemma}
For a Lagrangian $p$-form $F$, the functional $\delta_S(F)$ is also a Lagrangian $p$-form.
Moreover, the operator $\delta_S$ preserves supports.
\end{lemma}

\begin{proof}
Unraveling the definitions, one finds that
\[
\Qcal_S(F)(\phi) = \omega_{\Ecal^!}(\delta S|_\phi, \delta F|_\phi),
\]
which is a $p$-form on $X$.
This formula is manifestly local, as the derivatives $\delta S$ and $\delta F$ are local.
The claim about supports is identical to the proof for local functionals.
\end{proof}

By summing over $p$, 
we now have an operator of cohomological degree one on the whole collection of local differential forms,
and we denote it simply by~$\delta_S$.

\begin{prop}
If an action functional $S$ satisfies the classical master equation (i.e., $\Qcal_S^2 = 0$ on local functionals),
then 
\begin{itemize}
\item $\delta_S^2 = 0$ on Lagrangian $p$-forms and
\item $[\d_{\lloc}, \delta_S] = 0$ on $\lloc$.
\end{itemize}
In short, $(\lloc, \d_\lloc + \delta_S)$ is a cochain complex.
\end{prop}

A systematic treatment is found in \cite{Get21} at Proposition~2.3.

\begin{proof}
The first assertion is immediate: $\Qcal_S^2 = 0$ implies $L_{\Qcal_S} \circ L_{\Qcal_S} = 0$.
The second follows because these operators act, in a sense, on separate tensor factors of $\lloc$.
More explicitly, observe that on the product space $\Ecal \times X$, 
there is a sheaf of dg commutative rings $C^\infty_\Ecal \boxtimes \Omega^*_X$.
The operator $\Qcal_S$ only acts on the first factor, and the de Rham derivative only acts on the second factor,
and hence they manifestly commute.
\end{proof}

This construction extends without difficulty from $\lloc_c$ to $\mlloc_c$, 
the dg multilocal functionals.
Note that $\mlloc_c$ is a graded-commutative algebra with respect to the following multiplication.
Consider a functional $F$ homogeneous of degree $m$ so that
\[
F(\phi) = \int_{X^m} \alpha(\jet_k \phi(x_1),\ldots, \jet_k \phi(x_m)) \mu
\]
with $\alpha$ a function on the total space of a bundle ${\rm J^k} E^{[m]} \to X^m$ and $\mu$ a smooth density on $X^m$.
Let $G$ be homogeneous of degree $n$ so that
\[
G(\phi) = \int_{X^n} \beta(\jet_\ell \phi(x_1),\ldots, \jet_\ell \phi(x_m)) \nu
\]
with $\beta$ a function on the total space of a bundle ${\rm J^\ell} E^{[m]} \to X^m$ and $\nu$ a smooth density on $X^m$.
Then the product $F G$ is homogeneous of degree $m + n$ and is defined by
\[
(FG)(\phi) = \int_{X^{m+n}}\alpha(\jet_k \phi(x_1),\ldots, \jet_k \phi(x_m))  \beta(\jet_\ell \phi(x_{m+1}),\ldots, \jet_\ell \phi(x_{m+n})) \mu \wedge \nu.
\]
One extends $\delta_S$ to $\mlloc_c$ as a derivation with respect to this multiplication.
Concretely that means one applies $\delta_S$ along each factor of $X$ in a product space~$X^m$.

\subsection{The factorization algebra of classical observables from multilocal differential forms}

We now describe how to build a homotopy factorization algebra of classical observables by using the multilocal differential forms.

\begin{df}
For a solution $\mathcal{L}$ that satisfies the classical master equation~\eqref{eqn:cme}, 
let $\Obs^{cl}$ denote the prefactorization algebra on $X$ of {\em classical observables} for the theory.
It takes values in $\Ch$, and it assigns to each open set $U$, the cochain complex $(\mlloc_c(U), \d_\mlloc + \delta_S)$.
\end{df}

The key point is that the operator $\delta_S$ preserves the support of any Lagrangian $p$-forms and more generally of any multilocal differential form. 
Hence on each open set $U$,  it defines a deformation of $\mlloc_c(U)$ as a prefactorization algebra.
Moreover, the deformation intertwines with the structure maps of $\mlloc_c$:
for inclusions (i.e., unary maps), this is an immediate consequence of the support-preservation property, 
and it holds for the higher arity maps, which are defined by composing inclusions
(first into the disjoint union and then into the larger target open).

In fact, this deformation remains a homotopy factorization algebra.

\begin{prop}
The construction $\Obs^{cl}$ is a homotopy factorization algebra.
\end{prop}

\begin{proof}
For any open $U$, $\Obs^{cl}(U)$ can be seen as the totalization of a double complex.
Hence, for any Weiss cover $\Vcal$ of an open $V$, the \v{C}ech diagram can be seen as the totalization of a diagram of double complexes.
It maps to $\Obs^{cl}(V)$ by the structure maps.

Consider the spectral sequence for these double complexes that applies $\d_\mlloc$ first.
We have shown already that this \v{C}ech diagram is quasi-isomorphic to $\mlloc(V)$.
Hence as the map of spectral sequences becomes a quasi-isomorphism,
the original map from the \v{C}ech diagram for $\Obs^{cl}$ to $\Obs^{cl}(V)$ is also a quasi-isomorphism.
\end{proof}

\subsection{Fiberwise analyticity}
\label{analyticity of obscl}

It is often useful to restrict to multilocal functional or differential forms that are analytic functions along the fibers of jet bundles,
because most observables of interest to physicists are fiberwise analytic,
and because the renormalization procedure we use demands it.
Restricting the class of observables to being fiberwise analytic does, however, impose a restriction on the allowed dynamics:
we need the action functional $S$ to be fiberwise analytic as well, so that the cohomological vector field $\Qcal_S$ preserves fiberwise analyticity of multilocal functionals.
With that proviso, all the preceding work admits an immediate analog.

\begin{df}
For a fiberwise analytic Lagrangian density $\mathcal{L}$ that satisfies the classical master equation~\eqref{eqn:cme}, 
let $\Obs^{cl, f\omega}$ denote the prefactorization algebra on $X$ of {\em fiberwise analytic classical observables}.
It takes values in $\Ch$, and  it assigns to each open set $U$, the cochain complex $(\mlloc_c^{f\omega}(U), \d_\mlloc + \delta_S)$.
\end{df}

As fiberwise analyticity is a local condition along the spacetime manifold, 
the argument already given for all observables simply carries over.

\begin{prop}
The construction $\Obs^{cl, f\omega}$ is a homotopy factorization algebra.
\end{prop}

We will restrict to fiberwise analytic multilocal functionals in the next section, 
so we drop the superscript $f\omega$ in the remainder of the paper.

\section{The quantum observables}
\label{sec obsq}

In this final section, we recall what it means to quantize perturbatively following the Batalin-Vilkovisky formalism, as proposed in \cite{FR3}. Along the way we accrue results that let us conclude by showing that the quantum observables form a factorization algebra.

\subsection{Factorization algebras of regular polynomials for free quantum theories}
\label{review of GR}

The main result of \cite{GR20} is that a free quantum theory, as defined above, has a factorization algebra of observables and that it contains the information encoded by the theory's net of observables.
The construction in this paper employs different techniques, but it borrows some key ideas from our prior work,
so we quickly review what we need.

The key observation of that paper was that the relation between local nets and factorization algebras is established by means of the time-ordering operator $\TT$, which in the non-renormalized case is given by
\be\label{Tmap}
\TT=e^{\frac{1}{2}\partial_{G^{\rm F}}},
\ee
where $G^{\rm F}$ is the Feynman propagator and, for an integral kernel $G$, we use the notation $\partial_G$ to denote the differential operator 
\[
\partial_G F\doteq \iota_G(F^{(2)})\,,
\]
where $\iota_G(F^{(2)})(\ph)=\left<G,F^{(2)}(\ph)\right>$.

The map $\TT$ provides us with the cochain isomorphism
\[
\Acal
\xto{\TT} (PV[[\hbar]],\delta_{S})\,,
\]
where 
\[
\Acal(U) = (PV(U)[[\hbar]],\delta_{S}-i\hbar \triangle,\{-,-\}).
\]
We use $\delta_S$ for the classical BV differential of the free theory with action $S$ (denoted by $\d$ in \cite{GR20}), we use $\triangle$ for the BV Laplacian, and we use $PV$ for the space of polynomial polyvector fields, where degree $k$-polyvector field are defined by
\[
PV^{-k}(U) = \bigoplus_{n\geq 0} \Dcal_{n+k}(U)_{S_n \times S_k}\,.
\]
In the pAQFT language, $PV$ is understood as the space of regular polynomial smooth functionals on $\Ecal$ (in \cite{GR20} we use the notation $\mathfrak{PV}_{\pol}$ for these objects). It arises from  the pAQFT functor 
$$\mathfrak{A}_{\pol}(U)=(PV(U)[[\hbar]],\star,\delta_S)\big|_{{\bf Caus}(X)}\,,$$
which is defined in the first instance only on causally convex, relatively compact opens (giving rise to the subcategory ${\bf Caus}(X)$). Hence $(PV(U)[[\hbar]],\delta_S)|_{{\bf Caus}(X)}$ is obtained from $\mathfrak{A}_{\pol}$ after applying the forgetful functor to cochain complexes.

The key observation is that on the factorization algebra side we work with $PV$ equipped with the usual graded-commutative product $\cdot$, 
and on the pAQFT side, we use $\TT$ to deform it to the graded-commutative product $\T$ defined by
\be\label{Tprod}
F\T G\doteq \TT(\TT^{-1}F\cdot \TT^{-1} G)\,,
\ee
where $F,G\in PV$.
We can also use $\TT$ to deform the differential $\delta_{S}$ and introduce
\be\label{intertw}
\hat{s}=\TT^{-1}\circ \delta_{S} \circ \TT\,.
\ee
On $PV$, by explicit computation, one finds that
\[
\hat{s}=\delta_{S}-i\hbar \Lap\,,
\]
which is exactly the differential in the factorization algebra $\Acal$.
One can summarize the relation between the products and differentials on both sides of the comparison as follows:
\[
(\hat{s}=\delta_{S_0}-i\hbar \Lap,\cdot)\xto{\TT} (\delta_{S_0},\T)\,.
\]
On the left-hand side, the differential is deformed and the product remains as $\wedge$. On the right-hand side, the differential is left as $\delta_{S_0}$ and the product is deformed. The equivalence provided by $\TT$ means that both constructions lead to the same cochain complex of quantum observables.

\subsection{Factorization algebras of multilocal polynomials in free and interacting quantum theory}

In this paper we work with a larger class of observables, in order to treat interacting theories.
This change requires us to modify the approach from the free theories.

\subsubsection{Motivation and strategy}

The main result of this paper is a direct generalization of what we discussed in Section~\ref{review of GR}, i.e. the construction of factorization algebras of free and interacting quantum theories, with the use of Epstein-Glaser renormalisation.
 
Note that $\Tcal$ given by means of formula \eqref{Tmap} cannot be applied to local functionals that are at least quadratic in the fields, 
due to singularities of the Feynman propagator~$G^{\rm F}$. 
Hence we cannot  use the formulas from Subsection~\ref{review of GR} to extend the factorisation algebra of free field theory to general \emph{multilocal} observables, 
even in a free field theory.
Thus, we deploy an Epstein-Glaser renormalization process \cite{EG} on the pAQFT side and use a renormalized version of time-ordering $\TT$ given in Definition~\ref{df: T map} to {\em define} a factorization algebra.

We will see that this approach produces an interesting deformation of the factorization algebra $\Obs^{cl}$ of classical observables defined above, 
in which the structure maps are unchanged but the differential changes.
This deformation realizes one version of the slogan ``turn on the BV Laplacian and solve the quantum master equation.'' 
The resulting functor $\Obs^q_{\rm free}$ is constructed in Section~\ref{sec: qobs free}.


The next step is the generalization to interacting theories, already sketched in~\cite[8.1]{GR20}. 
Given a quadratic $S_0$ and the (higher than quadratic) interaction $V\in\loc_c(X)$, 
we define the renormalized interacting BV operator by
\be\label{def:qBVop:int}
\hat{s}(F):=e^{-i\lambda V/\hbar}\, \TT^{-1}\circ s_0 \circ \TT(e^{i\lambda V/\hbar} F)=e^{-i\lambda V/\hbar}\, \hat{s}_0 (e^{i\lambda V/\hbar} F)\,,
\ee
Assuming an extra renormalization condition called the \textit{quantum master equation} (QME) 
and using the \textit{master Ward identity}, proven by \cite{BreDue,H}, 
we will show that 
\[
\hat{s}_V=\delta_{S_0}+\{\cdot, V\}-i\hbar \Lap\,,
\]
where $\Lap$ is the renormalized BV Laplacian.

\begin{rem}
Given that the QME holds, we could also express $\hat{s}_V$ by means of the intertwining operator 
$$R_V(F)\doteq (e_{\TT}^{iV/\hbar})^{\star -1}\star(e_{\TT}^{iV/\hbar}\T F).$$
It was shown in \cite{FR3} that assuming the QME, we have
$$\hat{s}_V:=R_V^{-1}\circ\delta_{S_0}\circ R_V\,.$$
Hence, again, we obtain a deformation of the classical BV differential $\delta_{S_0}+\{\cdot, V\}$ through deformation of the product, this time using~$R_V$.
\end{rem}

\subsubsection{The anomalous master Ward identity}

We now recall a key result, known in the literature as the {\em anomalous master Ward identity} (AMWI).
It leads to an explicit formula for the renormalized $\hat{s}_0$  and the renormalized $\hat{s}_V$, also guaranteeing the locality of the latter.

\begin{thm}[AMWI]\label{thm: AMWI}
Let $L_0$ be a generalized Lagrangian and $f$ a test function. For any local compactly supported functional $F\in \loc_c$,
there is a local functional $A_{\lambda F}$, known as the {\em anomaly}, satisfying
\be
\label{MWI:inf:antif}
\delta_{L_0(f)}(\TT e^{i F/\hbar})=\tfrac{i}{\hbar}\TT\left(e^{i\lambda F/\hbar}\left(\tfrac{1}{2}\{L_0(f)+F,L_0(f)+F\}-i\hbar A_{F}\right)\right)\,,
\ee
the so-called {\em anomalous master Ward identity}. The anomaly is given in terms of maps $A_n:\lloc_c^{[n]}\rightarrow \overline{\mlloc}_{\mc}[[\hbar]]$, which are of order at least $\hbar^{n-1}$, and are local with the image supported on the small (or total) diagonal.  
For $F=\int_{X} \alpha$, where $\alpha$ is a top form in $\lloc_c^{[n]}$, we have 
\[
A_{F}=\sum_{n=1}^\infty \int_{X^n} A_n(\alpha^{\otimes n})\,.
\]
\end{thm}

The anomaly measures the deviation from the MWI.
Remarkably, this failure is itself local, which is not immediate from the definitions. 
In \cite{H} it was also proven that in good cases, 
one can remove the anomaly by redefining the time-ordered product.

Let $L_I$ be the interaction Lagrangian and $L=L_0+\lambda L_I$. 
Following \cite{FR3}, we formulate the {\it quantum master equation} (QME) as
\begin{equation}
\label{eq:QME}
\delta_{L_0(f)}(\TT e^{i\lambda L_I(f)/\hbar})=0\,,
\end{equation}
or, equivalently,
\begin{equation}\label{eq:dgQME2}
\tfrac{1}{2}\{ L(f), L(f)\}-i\hbar A_{\lambda L_I(f)}=0\,.
\end{equation}
The anomaly term $A_{\lambda\Lcal_I}$ is thus closely related to the renormalized BV Laplacian, 
so it is useful to introduce the following suggestive notation for its derivative:
\[
\Delta_{\lambda L_I(f)}(F)\doteq \frac{d}{d\lambda} A_{\lambda L_I(f)+\mu F}\Big|_{\mu=0}\,,
\]
where $F$ is a local functional.

Theorem~\ref{thm: AMWI} allows us to compute $\hat{s}_0\doteq \TT^{-1}\circ\, \delta_{L_0(f)}\circ \TT$ on multiolcal functionals of the form $F_1\cdot\ldots \cdot F_n$, where $F_1,\dots,F_n\in\loc_c$, 
using the formula
\[
F_1\cdot\ldots \cdot F_n=\left(\frac{\hbar}{i}\right)^n\frac{d^n}{d\lambda_1\dots d\lambda_n} e^{\frac{i}{\hbar}(\lambda_1F_1+\dots+\lambda_n F_n)}\big|_{\lambda_1=\dots \lambda_n=0}\,.
    \] 
Denote $f\equiv \lambda_1F_1+\dots+\lambda_n F_n $.
The result of the calculation is
    \[
            \hat{s}_0(F_1\cdot\ldots \cdot F_n)=\left(\frac{\hbar}{i}\right)^{n-1}\frac{d^n}{d\lambda_1\dots d\lambda_n} \left(e^{i F/\hbar}\left(\tfrac{1}{2}\{L_0(f)+F,L_0(f)+F\}-i\hbar A_{F}\right)\right)\big|_{\lambda_1=\dots \lambda_n=0}\,.
\]
This motivates the following definition:
\be\label{eq: Delta multiloc}
\Delta_0(F_1\cdot\ldots \cdot F_n)\doteq \left(\frac{\hbar}{i}\right)^{n-1}\frac{d^n}{d\lambda_1\dots d\lambda_n} \left(e^{i F/\hbar}\left(\tfrac{1}{2}\{F,F\}-i\hbar A_{F}\right)\right)\big|_{\lambda_1=\dots \lambda_n=0}\,,
\ee
which also generalizes to interacting theory as well; see formula (48) in~\cite{Rej13}. 
Notice a crucial property here:
if the $F_j$ have pairwise disjoint support, then
\[
\Delta_0(F_1\cdot\ldots \cdot F_n) = \sum_{j=1}^n \pm F_1 \cdot \ldots \Delta_0(F_j) \ldots \cdot F_n
\]
because $A_{F}$ vanishes as it is supported on the small diagonal.

\subsubsection{Quantum observables of the free theory}
\label{sec: qobs free}

In this section we construct the factorization algebra of quantum observables the free theory that contains polynomial multilocal functionals, 
extending the results of \cite{GR20}, where we only treated the case of regular polynomials. 
The essential idea is to replace the differential on the classical observables with $\hat{s}_0\doteq \TT^{-1}\circ\,  \delta_{S_0}\circ \TT$, where $\TT$ is given by Definition~\ref{df: T map}.

One must show that this new differential preserves supports to know that one still has a factorization algebra.
Using Theorem~\ref{thm: AMWI} and the definition of $\Delta_0$ in \eqref{eq: Delta multiloc} (which extends to general elements of $\mloc_c$ in a straightforward manner), 
we can write
\[
\hat{s}_0 F=\delta_{L_0(f)}F-i\hbar \Delta_0F=\delta_{S_0}F-i\hbar \Delta_0F\,,
\]
for $f$ chosen such that $f\equiv 1$ on~$\supp(F)$.
But these formulas are manifestly support-preserving, as the antibracket is and the anomaly is,
so we have a prefactorization algebra.

We now verify the local-to-global axiom.

\begin{prop}
\label{prop: obsqfree is fact alg}
The functor $\Obs_{\rm free}^{q}: \Opens(X) \to \Ch(\CC[[\hbar]])$, assigning 
\[
(\mloc_c(U)[[\hbar]], \hat{s}_0)
\]
to each open subset $U\subset X$, is a factorization algebra. 
\end{prop}

\begin{proof}
Fix a Weiss cover $\{U_i\}_i$ of any open subset $U \subset X$.
There is a canonical cochain map from the totalization of the \v{C}ech complex $C(\{U_i^m\}_i,\Obs_{\rm free}^q)$ to $\Obs_{\rm free}^q(U)$.
Now consider the filtration of the cochain complex above by powers of $\hbar$,
which exists on both the \v{C}ech complex and on $\Obs_{\rm free}^q(U)$ and hence determines a spectral sequence for both complexes.
As the cochain map preserves the filtration,
there is a map of spectral sequences.
This map is a quasi-isomorphism on the first page by Proposition~\ref{prop: obscl is a fact alg}, 
and hence the original map is also a quasi-isomorphism.
In other words, $\Obs_{\rm free}^q$ satisfies homotopy descent for Weiss covers.
\end{proof}

\subsubsection{Quantum observables of the interacting theory}
\label{sec:qobs}
We wish now to define the interacting quantum BV operator $\hat{s}$, 
following \cite{FR3}. We use formula \eqref{def:qBVop:int} with $\TT$ given in Defintion~\ref{df: T map}. Locality of $\hat{s}$ is
 an obvious corollary of Theorem~\ref{thm: AMWI}.

\begin{cor} 
Let $L=L_0+L_I$ be the Lagrangian of the theory. If the QME holds for $L(f)$ with appropriate choice of the test function(s) $f$, the quantum BV operator is local, support-preserving, and is given~by
\be
\label{eq:qBV}
\hat{s}(F)=\{S,F\}-i\hbar \Delta_{\lambda L_I(f)}(F)\,,
\ee
provided $f\equiv 1$ on $\supp(F)$.
\end{cor}

\begin{proof}
Let $F\in \loc_c$. We apply the AMWI (Theorem~\ref{thm: AMWI}) to conclude that
\begin{align*}
\hat{s}(F)
&=(-i\hbar)\frac{d}{d\lambda}e^{-i\lambda L_I(f)/\hbar} \delta_{L_0(f)}(e^{i(\lambda L_I(f)+\lambda F)/\hbar} )\Big|_{\lambda=0}\\
&=(-i\hbar)\frac{d}{d\lambda} \Big(e^{i(\lambda L(f)+F)/\hbar} \big(\tfrac{1}{2}\{L(f)+\lambda F,L(f)+ F\}-i\hbar A_{\lambda L_I(f)+\lambda F }\big)\Big)\Big|_{\lambda=0}\\	&=\frac{i}{\hbar}F\left(\tfrac{1}{2}\{L_I(f),L_I(f)\}-i\hbar A_{\lambda L_I(f)}\right)+ \{L_I(f),F\}-i\hbar \Delta_{\lambda L(f)}(F)\,.
\end{align*} 
The first term in the last equation vanishes if the QME holds, and the second term gives the desired form~\eqref{eq:qBV} of $\hat{s}$, where we also use the fact that $f\equiv 1$ on the support of $F$.
\end{proof}

\begin{cor}
If the QME holds, the functor $\Obs^q$ assigning 
\[
(\mloc_c(U)[[\hbar,\lambda]], \hat{s})
\] 
to each open subset $U \subset M$ is a factorization algebra. 
\end{cor}

The same proof as for Proposition~\ref{prop: obsqfree is fact alg} works here,
using the spectral sequence on the \v{C}ech complex associated to the $\hbar$ filtration.

\subsection{Why examples exist}\label{sec: examples}

It is natural and pertinent to ask if there are interesting examples of our main result.
We have seen that the QME is the necessary and sufficient condition for the locality of~$\hat{s}$. The following examples have been discussed in the literature:
\begin{itemize}
    \item Single real scalar field. In this case the CME is trivially satisfied and the anomaly vanishes\cite{FR3}.
    \item $n$ real scalar fields has been discussed in \cite{brunetti2022unitary}.
    \item Yang-Mills theory \cite{H,FR3}
    \item Bosonic string \cite{BRZ}.
    \item Effective quantum gravity \cite{BFRej13}
\end{itemize}
These references were written at different stages of the development of the whole framework, so they contained a variable level of detail and notation might differ significantly. To make it easier for the reader to navigate the literature, here we provide the core argument used in proving QME in all these cases. This strategy applies even if given papers do not state it explicitly.

\begin{thm}\label{thm: existence}
    Assume that for the given Lagrangian density there exists a prescription for $L(f)$ such that \eqref{eq: CME strong} holds. Assume further that  $H^1(\lloc(M),\d_\mlloc + \delta_S)$ is trivial. Then the QME can be fulfilled through appropriate redefinition of time-ordered product and the functor $\Obs^q$ produces a factorisation algebra for the model.
\end{thm}
\begin{proof}
In \cite{r00095,H,FR3} the following strategy to prove the QME has been outlined. One assumes first that
\be\label{eq: CME strong}
\{ L(f), L(f)\}=0\,,
\ee
which also implies that $\delta_{L(f)}^2=0$. (This in turn implies $\delta_{S}^2=0$, but the implication does not go the other way.) Then one shows that the anomaly in \eqref{MWI:inf:antif} can be removed by appropriate redefinition of~$\Tcal$. 
This step amounts to solving a certain cohomological condition.

The key observation is that the anomaly in AMWI has to satisfy the \emph{Wess-Zumino consistency condition}. This fact was first identified  in \cite{H} and then generalized in \cite{FR3}. The detailed argument explaining why this condition implies that anomalies are classified by   $H^1(\lloc(M),\d_\mlloc + \delta_S)$ is provided in \cite{r00095}.
\end{proof}

\begin{exa}
We demonstrate here that Yang-Mills theory without chiral fermions satisfies the assumptions of Theorem~\ref{thm: existence}.

In Yang-Mills theory, it is convenient to use the following choice of test functions: 
$\euD=\cinfty_c(X,\RR^2)$, so $f=(f',f'')$ 
and we choose $f'\equiv 1$ on the support of $f''$. 
Let $\frak{g}$ denote the Lie algebra of a compact Lie group~$G$.
Before the gauge-fixing, we use the following generalized Lagrangian 
(written in local coordinates and a given basis for~$\frak{g}$):
\begin{multline*}
L(f)=-\frac{1}{2}\int_X \tr(F[f'A]\wedge *F[f'A])+\int_X \big(d(f''c)+\frac{1}{2}[f'A,f''c]\big)^I_\mu(x)\frac{\delta}{\delta (f'A^I_\mu(x))}\\+\frac{1}{2}\int_X  [f''c,f''c]^I(x)\frac{\delta}{\delta (f''c^I(x))}-i\int_X  f''b_I(x)\frac{\delta}{\delta (f'\bar{c}_I(x))}\,,
\end{multline*}
where $A$ is the vector potential in $\Omega_1(X,\frak{g})$, $c\in\cinfty(X,\frak{g})$ is the ghost (which has degree 1), 
$\bar{c}\in\cinfty(X,\frak{g})$ is the antighosts in degree -1, 
and
$b \in\cinfty(X,\frak{g})$ is the Nakanishi-Lautrup field. 
Here $F[A]$ is the field strength for the vector potential $A$ and $F[f'A]$ is that, but for $A$ multiplied with the test function $f'$.

Our choice of dependence on $f$ and the fact that $f'\equiv 1$ on $\supp f''$ guarantees that we can effectively use the same algebraic relation one would apply for compactly supported fields, 
so the standard calculation proves that $\{L(f),L(f)\}=0$. 
Note that for consistency we have ``regularized'' not only fields, 
but also antifields, which is then taken into account when computing the bracket. 

To obtain the gauge-fixed theory (i.e. to guarantee that the linearized equations of motion in antifield degree 0 give rise to normally hyperbolic equations), 
we introduce a gauge-fixing fermion:
\[
\Psi(f)=i\int\limits_X f'\bar{c}_I\left(\frac{1}{2}f''b^I+*^{-1}d*\!f'A^I\right)\dvol\,.
\]
The gauge-fixed action is $L(f)+\{L(f),\Psi(f)\}$, 
which also satisfies the CME.

The cohomological condition is also fulfilled for pure Yang-Mills, as argued in \cite{H}. In section 5.5 of \cite{H}, constraints on adding possible matter fields are also discussed. In particular, the presence of chiral fermions would introduce an anomaly that cannot be removed from the QME, so the standard argument presented above would no appy.
\end{exa}


\bibliographystyle{amsalpha}
\bibliography{References}

\newcommand{\etalchar}[1]{$^{#1}$}
\providecommand{\bysame}{\leavevmode\hbox to3em{\hrulefill}\thinspace}
\providecommand{\MR}{\relax\ifhmode\unskip\space\fi MR }
\providecommand{\MRhref}[2]{%
  \href{http://www.ams.org/mathscinet-getitem?mr=#1}{#2}
}
\providecommand{\href}[2]{#2}
\begin{thebibliography}{BDLGR18}

\bibitem[B{\"a}r15]{GreenBear}
C.~B{\"a}r, \emph{Green-hyperbolic operators on globally hyperbolic
  spacetimes}, Communications in Mathematical Physics \textbf{333} (2015),
  no.~3, 1585--1615.

\bibitem[BBH00]{r00095}
Glenn Barnich, Friedemann Brandt, and Marc Henneaux, \emph{Local {BRST}
  cohomology in gauge theories}, Phys. Rep. \textbf{338} (2000), no.~5,
  439--569. \MR{1792979}

\bibitem[BD04]{BD}
Alexander Beilinson and Vladimir Drinfeld, \emph{Chiral algebras}, American
  Mathematical Society Colloquium Publications, vol.~51, American Mathematical
  Society, Providence, RI, 2004. \MR{2058353}

\bibitem[BD08]{BreDue}
F.~Brennecke and M.~D{\"u}tsch, \emph{Removal of violations of the master
  {W}ard identity in perturbative {QFT}}, Reviews in Mathematical Physics
  \textbf{20} (2008), no.~02, 119--51.

\bibitem[BDF09]{BDF}
R.~Brunetti, M.~D{\"u}tsch, and K.~Fredenhagen, \emph{Perturbative algebraic
  quantum field theory and the renormalization groups}, Adv. Theor. Math. Phys.
  \textbf{13} (2009), no.~5, 1541--1599.

\bibitem[BDFR22]{brunetti2022unitary}
Romeo Brunetti, Michael D{\"u}tsch, Klaus Fredenhagen, and Kasia Rejzner,
  \emph{Unitary, anomalous master ward identity and its connections to the
  wess-zumino condition, bv formalism and $l_\infty$-algebras}, arXiv preprint
  arXiv:2210.05908 (2022).

\bibitem[BDLGR18]{BDGR}
Christian Brouder, Nguyen~Viet Dang, Camille Laurent-Gengoux, and Kasia
  Rejzner, \emph{Properties of field functionals and characterization of local
  functionals}, J. Math. Phys. \textbf{59} (2018), no.~2, 023508, 47.
  \MR{3765742}

\bibitem[BF97]{BF97}
R.~Brunetti and K.~Fredenhagen, \emph{Interacting quantum fields in curved
  space: renormalizability of {$\phi^4$}}, Operator algebras and quantum field
  theory ({R}ome, 1996), Int. Press, Cambridge, MA, 1997, pp.~546--563.
  \MR{1491141}

\bibitem[BF00]{BF0}
\bysame, \emph{Microlocal analysis and interacting quantum field theories},
  Commun. Math. Phys. \textbf{208} (2000), no.~3, 623--661.

\bibitem[BF09]{BFBook}
\bysame, \emph{Quantum field theory on curved backgrounds}, pp.~129--155,
  Springer, 2009.

\bibitem[BFK96]{BFK96}
R.~Brunetti, K.~Fredenhagen, and M.~K{\"o}hler, \emph{The microlocal spectrum
  condition and {W}ick polynomials of free fields on curved spacetimes},
  Commun. Math. Phys. \textbf{180} (1996), no.~3, 633--652.

\bibitem[BFR16]{BFRej13}
R.~Brunetti, K.~Fredenhagen, and K.~Rejzner, \emph{Quantum gravity from the
  point of view of locally covariant quantum field theory}, Communications in
  Mathematical Physics \textbf{345} (2016), no.~3, 741--779.

\bibitem[BFR19]{BFR}
R.~Brunetti, K.~Fredenhagen, and P.~L. Ribeiro, \emph{Algebraic structure of
  classical field theory: Kinematics and linearized dynamics for real scalar
  fields}, Commun. Math. Phys. \textbf{368} (2019), 519--584.

\bibitem[BFV03]{BFV}
R.~Brunetti, K.~Fredenhagen, and R.~Verch, \emph{The generally covariant
  locality principle---{A} new paradigm for local quantum field theory},
  Commun. Math. Phys. \textbf{237} (2003), 31--68.

\bibitem[BH96]{BarHen}
G.~Barnich and M.~Henneaux, \emph{Isomorphisms between the
  {Batalin–Vilkovisky} antibracket and the {Poisson} bracket}, Journal of
  Mathematical Physics \textbf{37} (1996), no.~11, 5273--5296.

\bibitem[BPS20]{BPS19}
M.~Benini, M.~Perin, and A.~Schenkel, \emph{Model-independent comparison
  between factorization algebras and algebraic quantum field theory on
  lorentzian manifolds}, Communications in Mathematical Physics \textbf{377}
  (2020), 971–997.

\bibitem[Bre97]{Bredon}
Glen~E. Bredon, \emph{Sheaf theory}, second ed., Graduate Texts in Mathematics,
  vol. 170, Springer-Verlag, New York, 1997. \MR{1481706}

\bibitem[BRZ14]{BRZ}
D.~Bahns, K.~Rejzner, and J.~Zahn, \emph{The effective theory of strings},
  Communications in Mathematical Physics \textbf{327} (2014), no.~3, 779--814.

\bibitem[BS17]{BS17}
M.~Benini and A.~Schenkel, \emph{Quantum field theories on categories fibered
  in groupoids}, Communications in Mathematical Physics \textbf{356} (2017),
  no.~1, 19--64.

\bibitem[BS19]{BS19}
S.~Bruinsma and A.~Schenkel, \emph{Algebraic field theory operads and linear
  quantization}, Letters in Mathematical Physics \textbf{109} (2019), no.~11,
  2531--2570.

\bibitem[BSS18]{BSS17}
M.~Benini, A.~Schenkel, and U.~Schreiber, \emph{The stack of {Y}ang-{M}ills
  fields on {L}orentzian manifolds}, Communications in Mathematical Physics
  \textbf{359} (2018), no.~2, 765--820.

\bibitem[BSW19]{BSW19}
M.~Benini, A.~Schenkel, and L.~Woike, \emph{Homotopy theory of algebraic
  quantum field theories}, Letters in Mathematical Physics \textbf{109} (2019),
  no.~7, 1487--1532.

\bibitem[BT13]{BottTu}
R.~Bott and L.~W. Tu, \emph{Differential forms in algebraic topology}, vol.~82,
  Springer Science and Business Media, 2013.

\bibitem[CG17]{CoGw}
K.~Costello and O.~Gwilliam, \emph{Factorization algebras in perturbative
  quantum field theory. {V}ol. 1}, New Mathematical Monographs, vol.~31,
  Cambridge University Press, Cambridge, 2017,
  \url{http://people.mpim-bonn.mpg.de/gwilliam/vol1may8.pdf}.

\bibitem[CG21]{CG2}
Kevin Costello and Owen Gwilliam, \emph{Factorization algebras in quantum field
  theory. {V}ol. 2}, New Mathematical Monographs, vol.~41, Cambridge University
  Press, Cambridge, 2021. \MR{4300181}

\bibitem[CMR14]{CMR1}
A.~S. Cattaneo, P.~Mnev, and N.~Reshetikhin, \emph{Classical {BV} theories on
  manifolds with boundary}, Communications in Mathematical Physics \textbf{332}
  (2014), no.~2, 535--603, CMR1.

\bibitem[CMR18]{CMR2}
\bysame, \emph{Perturbative quantum gauge theories on manifolds with boundary},
  Communications in Mathematical Physics \textbf{357} (2018), no.~2, 631--730.

\bibitem[Cos11]{Cos}
K.~Costello, \emph{Renormalization and effective field theory}, Mathematical
  Surveys and Monographs, AMS, Providence, Rhode Island, 2011.

\bibitem[Dan13]{Viet}
N.~V. Dang, \emph{Renormalization of quantum field theory on curved
  space-times, a causal approach}, [arXiv:1312.5674] (2013).

\bibitem[DEF{\etalchar{+}}99]{FieldsStrings}
P.~Deligne, P.~I. Etingof, D.~S. Freed, L.~C. Jeffrey, D.~Kazhdan, J.~W.
  Morgan, D.~A. Morrison, and E.~Witten, \emph{Quantum fields and strings: a
  course for mathematicians}, vol. 1 and 2, American Mathematical Society
  Providence, 1999.

\bibitem[DF99]{DFqed}
M.~Dütsch and K.~Fredenhagen, \emph{A local (perturbative) construction of
  observables in gauge theories: {T}he example of {QED}}, Communications in
  Mathematical Physics \textbf{203} (1999), no.~1, 71--105.

\bibitem[DF01]{DF}
M.~D{\"u}tsch and K.~Fredenhagen, \emph{Perturbative algebraic field theory,
  and deformation quantization}, Mathematical Physics in Mathematics and
  Physics: Quantum and Operator Algebraic Aspects \textbf{30} (2001), 151--160.

\bibitem[DF04]{DF04}
\bysame, \emph{Causal perturbation theory in terms of retarded products, and a
  proof of the {A}ction {W}ard {I}dentity}, Reviews in Mathematical Physics
  \textbf{16} (2004), no.~10, 1291--1348.

\bibitem[DF07]{DF05}
\bysame, \emph{Action {W}ard identity and the {S}t{\"u}ckelberg-{P}etermann
  renormalization group}, Progress in Mathematics, vol. 251, pp.~113--123,
  Birkh{\"a}user Verlag, Basel, 2007.

\bibitem[D{\"u}t19]{Due19}
M.~D{\"u}tsch, \emph{From classical field theory to perturbative quantum field
  theory}, Birkh{\"a}user, 2019.

\bibitem[EG73]{EG}
H.~Epstein and V.~Glaser, \emph{The role of locality in perturbation theory},
  AHP \textbf{19} (1973), no.~3, 211--295.

\bibitem[FR12]{FR}
K.~Fredenhagen and K.~Rejzner, \emph{{B}atalin-{V}ilkovisky formalism in the
  functional approach to classical field theory}, Communications in
  Mathematical Physics \textbf{314} (2012), no.~1, 93--127.

\bibitem[FR13]{FR3}
\bysame, \emph{{B}atalin-{V}ilkovisky formalism in perturbative algebraic
  quantum field theory}, Communications in Mathematical Physics \textbf{317}
  (2013), no.~3, 697--725.

\bibitem[Get21]{Get21}
Ezra Getzler, \emph{Batalin-{V}ilkovisky formality for {C}hern-{S}imons
  theory}, J. High Energy Phys. (2021), no.~12, Paper No. 105, 23. \MR{4370868}

\bibitem[GKSW15]{gaiotto2015generalized}
Davide Gaiotto, Anton Kapustin, Nathan Seiberg, and Brian Willett,
  \emph{Generalized global symmetries}, Journal of High Energy Physics
  \textbf{2015} (2015), no.~2, 1--62.

\bibitem[GR20]{GR20}
O.~Gwilliam and K.~Rejzner, \emph{Relating nets and factorization algebras of
  observables: Free field theories}, Communications in Mathematical Physics
  \textbf{373} (2020), no.~1, 107--174.

\bibitem[Hol08]{H}
S.~Hollands, \emph{Renormalized quantum {Y}ang-{M}ills fields in curved
  spacetime}, Reviews in Mathematical Physics \textbf{20} (2008), 1033--1172,
  [arXiv:gr-qc/705.3340v3].

\bibitem[HW01]{HW}
S.~Hollands and R.~M. Wald, \emph{Local {W}ick polynomials and time ordered
  products of quantum fields in curved spacetime}, Commun. Math. Phys.
  \textbf{223} (2001), no.~2, 289--326.

\bibitem[HW02a]{HW01}
\bysame, \emph{Existence of local covariant time ordered products of quantum
  fields in curved spacetime}, Commun. Math. Phys. \textbf{231} (2002), no.~2,
  309--345.

\bibitem[HW02b]{HW02}
\bysame, \emph{On the renormalization group in curved spacetime}, Commun. Math.
  Phys. \textbf{237} (2002), 123--160.

\bibitem[KM97]{Michor}
Andreas Kriegl and Peter~W. Michor, \emph{The convenient setting of global
  analysis}, Mathematical Surveys and Monographs, vol.~53, American
  Mathematical Society, Providence, RI, 1997. \MR{1471480}

\bibitem[MSW20]{mnev2020towards}
Pavel Mnev, Michele Schiavina, and Konstantin Wernli, \emph{Towards holography
  in the {BV-BFV} setting}, Annales Henri Poincare \textbf{21} (2020), no.~3,
  993--1044.

\bibitem[Olv93]{Olver}
Peter~J. Olver, \emph{Applications of {L}ie groups to differential equations},
  second ed., Graduate Texts in Mathematics, vol. 107, Springer-Verlag, New
  York, 1993. \MR{1240056}

\bibitem[PS53]{SP53}
A.~Petermann and E.~St{\"u}ckelberg, \emph{La normalisation des constantes dans
  la th{\'e}orie des quanta}, Helv. Phys. Acta \textbf{26} (1953), 499--520.

\bibitem[PS16]{SP82}
G.~Popineau and R.~Stora, \emph{A pedagogical remark on the main theorem of
  perturbative renormalization theory}, Nuclear Physics B \textbf{912} (2016),
  70--78.

\bibitem[Rej14]{Rej13}
K.~Rejzner, \emph{Remarks on local symmetry invariance in perturbative
  algebraic quantum field theory}, Annales Henri Poincar{\'e} \textbf{16}
  (2014), no.~1, 205--238.

\bibitem[Rej16]{Book}
\bysame, \emph{Perturbative {A}lgebraic {Q}uantum {F}ield {T}heory. {A}n
  introduction for {M}athematicians}, Mathematical Physics Studies, Springer,
  2016.

\bibitem[Rej19]{Rej19}
\bysame, \emph{Locality and causality in perturbative algebraic quantum field
  theory}, Journal of Mathematical Physics \textbf{60} (2019), no.~12, 122301.

\end{thebibliography}

\end{document}